\documentclass[10pt,conference]{IEEEtran}
\makeatletter
\def\ps@headings{%
\def\@oddhead{\mbox{}\scriptsize\rightmark \hfil \thepage}%
\def\@evenhead{\scriptsize\thepage \hfil\leftmark\mbox{}}%
\def\@oddfoot{}%
\def\@evenfoot{}}
\makeatother 

\usepackage[stable]{footmisc}
\usepackage[long]{optional}

\usepackage{cite}
   \usepackage[dvips]{graphicx}

\usepackage[cmex10]{amsmath}
\usepackage{algorithmicx}
\usepackage[ruled]{algorithm}
\usepackage{array}
\usepackage[tight,footnotesize]{subfigure}

\usepackage{graphicx, amssymb,subfigure,float}
\usepackage{times}
\usepackage[section]{placeins}

\newtheorem{theorem}{Theorem}
\newtheorem{lemma}{Lemma}
\newtheorem{definition}{Definition}

\def\squareforqed{\hbox{\rlap{$\sqcap$}$\sqcup$}}

\def\qed{\ifmmode\squareforqed\else{\nobreak\hfil
\penalty50\hskip1em\null\nobreak\hfil\squareforqed
\parfillskip=0pt\finalhyphendemerits=0\endgraf}\fi}

\IEEEoverridecommandlockouts

\begin{document}

\title{Optimal CSMA-based Wireless Communication with Worst-case Delay and Non-uniform Sizes\thanks{This research is supported in part by National Science Foundation awards CNS-11-17539. Any opinions, findings, and conclusions or recommendations expressed here are those of the authors and do not necessarily reflect the views of the funding agencies or the U.S. government.}\vspace{-3mm}}

\author{\IEEEauthorblockN{
Hongxing Li and Nitin Vaidya}\vspace{-3mm}\\
\IEEEauthorblockA{Coordinated Science Laboratory, University of Illinois at Urbana-Champaign, USA\\
Email: \{hxli, nhv\}@illinois.edu}\vspace{-8mm}
}

% make the title area
\maketitle

\begin{abstract}
Carrier Sense Multiple Access (CSMA) protocols have been shown to reach the full capacity region for data communication in wireless networks, with polynomial complexity. However, current literature achieves the throughput optimality with an exponential delay scaling with the network size, even in a simplified scenario for transmission jobs with uniform sizes. Although CSMA protocols with order-optimal average delay have been proposed for specific topologies, no existing work can provide \emph{worst-case delay guarantee} for each job in \emph{general network settings}, not to mention the case when the jobs have \emph{non-uniform lengths} while the throughput optimality is still targeted. In this paper, we tackle on this issue by proposing a two-timescale CSMA-based data communication protocol with dynamic decisions on rate control, link scheduling, job transmission and dropping in polynomial complexity. Through rigorous analysis, we demonstrate that the proposed protocol can achieve a throughput utility arbitrarily close to its offline optima for jobs with non-uniform sizes and worst-case delay guarantees, with a tradeoff of longer maximum allowable delay.

%Extensive simulation studies further examine the effectiveness and the delay-optimality tradeoff of our protocol.
\end{abstract}

\section{Introduction}\label{sec:intro}%\vspace{-1mm}
The efficacy of a wireless communication algorithm can be examined with three criteria: high throughput, low response delay and low computation/communication complexity. However, it is commonly accepted that there is a tradeoff among the three dimensions of algorithm performances \cite{shah-itw10}. Maximum-weight scheduling (MWS) \cite{tassiulas-tac92} algorithms are proven to be throughput-optimal, however incurring exponential computation complexity as the network size grows up. Low-complexity algorithms (\cite{lin-jsac06} and references therein) are proposed to approximate the MWS, while achieving only a fraction of the optimal throughput.

CSMA-style random access control protocols have been studied intensely in recent years for its low complexity and provable optimality in throughput maximization \cite{jiang-ton10, ni-infocom10}. Nevertheless, it comes with an exponentially long delay scaling with the network size \cite{mahdi-infocom11}. Although some recent efforts \cite{shah-sigmetrics10, subram-isit11, jiang-ita10, jiang-infocom11, mahdi-infocom11, lam-isit12, huang-infocom13} try to improve the delay performance and have even achieved asymptotic bounds on the average delay \cite{shah-sigmetrics10, subram-isit11, mahdi-infocom11} in specific topologies, the worst-case delay guarantee, which is a more practical concern in real-world implementations\footnote{It is especially true for delay-constrained applications including video streaming.} ensuring that each transmission job is either served or dropped before its maximum allowable delay, is yet to be studied. The difficulty further escalates if we still aim to obtain (close-to-)optimal throughput utility at the same time, with general network topologies and low computation/communication complexities.

Apart from above, a common assumption is shared by current literature such that each transmission job has the same size and is packed in a single data unit, \emph{e.g.}, one data packet, which can be completely delivered within one time slot. This idealized model fails to capture the diverse job sizes of some mainstream applications. For example, one \emph{Twitter} update may need just tens of bytes while a video clip on \emph{Youtube} may be in the size of several mega-bytes. A more practical model should allow the existence of transmission jobs consisted of one/multiple consecutive data packets, which should either be fully delivered to the destination or completely dropped. Partial reception of the transmission job brings no utility to the network, \emph{e.g.}, a video clip with missing information may not be decodable. When coupled with the worst-case delay guarantee, \emph{i.e.}, each transmission job instead of one packet is either delivered or dropped before its service deadline, we should explore novel designs for the low-complexity throughput-optimal CSMA protocol.

In this paper, we investigate the throughput-utility optimal CSMA protocol in general network topologies with low computation/communication complexities and worst-case delay guarantees for transmission jobs with diverse sizes. A two-timescale algorithm is proposed to dynamically make decisions in each time slot on: 1) \emph{rate control}: how many jobs should be admitted into the network such that congestion could be avoided while the throughput utility is maximized? 2) \emph{link scheduling}: which subset of the links should be simultaneously scheduled for transmission such that no collision will occur while the network capacity can be fully exploited? 3) \emph{job transmission}:  how many jobs, from each category of job sizes and worst-case delay requirements, should be transmitted over the scheduled links? 4) \emph{job dropping}: how many jobs of each category should be dropped so as to meet the worst-case delay bounds? A CSMA-style random access control mechanism is integrated with the Lyapunov optimization framework \cite{book2010} for the algorithm design. To be specific, the link scheduling is carried out with the CSMA protocol and \emph{randomly }generates collision-free transmissions, while the rate control, and job transmission and dropping decisions are \emph{deterministically} made based on the network status in each time slot. Rigorous analysis demonstrates that our protocol can achieve a throughput utility, which can be made arbitrarily close to its optima, with polynomial computation/communication complexity at each link and guaranteed worst-case delay for jobs with non-uniform sizes, at a tradeoff of longer maximum allowable delay.

The contribution of this paper is summarized as follows,

\vspace{1mm}
\noindent $\triangleright$ To our best knowledge, we are the first to investigate the existence of worst-case delay guarantees and non-uniform job sizes for CSMA protocols in general network topologies.

\vspace{1mm}
\noindent $\triangleright$ A CSMA-based two-timescale wireless communication algorithm is proposed to dynamically decide the rate control, link scheduling, and job transmission and dropping in each time slot, with an objective to maximize the time-averaged throughput utility.

\vspace{1mm}
\noindent $\triangleright$ Theoretical analysis demonstrates that our proposed algorithm can guarantee the worst-case delay for all job sizes, and achieve a throughput-utility that can be arbitrarily close to its optimality, with a polynomial computation/communication complexity at each link and the tradeoff of a longer maximum allowable delay.

%Extensive simulation studies further evaluate the efficiency and the optimality-delay tradeoff of our algorithm.

The rest of the paper is organized as follows. We discuss the related works in Sec.~\ref{sec:relatedwork} and present the problem model in Sec.~\ref{sec:model}. The two-timescale dynamic algorithm is introduced with details in Sec.~\ref{sec:algorithm}. In Sec.~\ref{sec:analysis}, we rigorously analyze the efficiency of our algorithm. Finally, we conclude this paper in Sec.~\ref{sec:conclusion}.

\section{Related Work}\label{sec:relatedwork}%\vspace{-1mm}
CSMA protocols have attracted tremendous attention in recent years \cite{jiang-ton10, ni-infocom10} mainly due to its potential to simultaneously achieve high throughput and low complexity, and its implementation in a distributed fashion.

With perfect and instantaneous carrier sensing assumption (no collision will happen), Jiang \emph{et al.~}\cite{jiang-ton10} introduce a continuous-time CSMA protocol, that can achieve the optimal throughput. A discrete-time queue-length based CSMA protocol is next proposed by Ni \emph{et al.~}\cite{ni-infocom10} to reach the full capacity region, explicitly considering the avoidance of collisions without the perfect carrier sensing assumption.

However, it has been shown that it is hard to achieve throughput optimality and low delay simultaneously with CSMA protocols \cite{shah-itw10}. Hence, a rich body of research efforts have been devoted to decreasing the delay for CSMA protocols. Shah \emph{et al.~}\cite{shah-sigmetrics10} present a CSMA protocol with order-optimal delay for networks with geometry. Jiang \emph{et al.~}\cite{jiang-ita10} demonstrate the relation between the small mixing time and a low delay, and investigate on how to tighten the generic bound of mixing time for specific topologies. For networks with bounded interference degree and an arrival rate within only a fraction of the capacity region, Jiang \emph{et al.~}\cite{jiang-infocom11} show that the average delay grows polynomially with the network size under parallel Glauber dynamics, and a constantly bounded mean delay independent of the network size is proved by Subramanian \emph{et al.~}\cite{subram-isit11}. In contrast, Lotfinezhad \emph{et al.~}\cite{mahdi-infocom11} achieve not only the throughput optimality but also an order-optimal delay, in the torus topology. However, it is not clear whether the above improvements can be extended to \emph{general network topologies}.

%, for which this paper presents a feasible solution.

Lee \emph{et al.}~\cite{lee-isit12} examine the delay performance of a class of CSMA protocol by tuning the control parameters. Nevertheless, there is no evidence on whether tuning parameters could fundamentally improve the exponential order of delay performance. Lam \emph{et al.~}\cite{lam-isit12} try to improve the average delay with multiple physical channels, however, each link can be scheduled on at most one channel at a time, which cannot fully exploit the capacity region. Huang \emph{et al.~}\cite{huang-infocom13} explore the power of multiple virtual channels to reduce the head-of-line delay, with a definition different from the average delay. Throughput utility, instead of queue lengthes, is used as the scheduling weight. Order-optimal head-of-line delay can be obtained in a close-loop setting with rate control.

The only work that considers the deadline for data transmission by Li \emph{et al.~}in \cite{li-rawnet11}. However, the solution proposed by \cite{li-rawnet11} only applies to a complete graph (each link collides with each other), but cannot be adapted to general topologies.

Different from the papers discussed above, this paper practically considers the guarantee of \emph{worst-case delay} bounds and the communication service for jobs with \emph{non-uniform sizes}. Meanwhile, the throughput-utility optimality can still be achieved with a \emph{low complexity}, for \emph{general topologies}.

%two-timescale \cite{theja-infocom12}

%time separation assumption, commonly assumed in, justified in \cite{jiang-tech08, liu-tech08} 

\section{Problem Model}\label{sec:model}%\vspace{-1mm}
We have a wireless network composed of a node set $\mathcal{N}$ and link set $\mathcal{E}$. Each source-destination pair is within one hop distance, which means each source just needs exactly one transmission to reach its destination without relaying. Each link has unit-capacity, \emph{i.e.}, transmitting at most one packet in one time slot.

We consider a general interference model by defining an interference-relation set $\mathcal{C}_i$ for each link $i\in \mathcal{E}$. Each link $j\in \mathcal{C}_i$ will cause collision to link $i$ scheduled concurrently.

%This general interference model can include the node-exclusive interfere and $k-$hop interference model, which are widely used in literature (\textbf{add citation}).

Different from existing efforts on throughput-optimal CSMA protocols assuming identical sizes of transmission jobs, we model the diverse job sizes of various network applications by differentiating types of transmission jobs. Let $\mathcal{M}$ denote the set of job types. For each job type $m\in \mathcal{M}$, it is composed of $s_m$ consecutive data packets, which should be either fully delivered to its destination or entirely dropped.

 %before its response deadline $D_m$. This deadline is counted since the arrival of this job.

%That is,
%
%\vspace{-4mm}{\small
%\begin{align}\label{eqn:bundle}
%&\text{For each job of type $m$ on link $i$, all its $s_m$ packets should be} \notag \\
%&\text{either fully delivered or completed dropped, }\forall m\in \mathcal{M}, i\in \mathcal{E}.
%\end{align}}\vspace{-4mm}

The network runs in a time-slotted fashion. In each time slot $t\geq 0$, a random number of $A_{mi}(t)$ ($\forall m\in \mathcal{M}, i\in \mathcal{E}$) jobs arrive at the transmitter of link $i$. Here, $A_{mi}(t)$ is i.i.d. in $[0, A_{mi}^{max}]$ with $A_{mi}^{max}$ as the maximum job arrival rate for type $m$ job at link $i$. Uncontrolled admission of job arrivals may cause congestion in the network. Thus, a rate control decision $r_{mi}(t)$ should be made such that jobs of type-$m$ are admitted into the job queue on link $i$ with

\vspace{-4mm}{\small
\begin{align}\label{eqn:rate-cons}
r_{mi}(t)\in [0, A_{mi}(t)],\ \forall m\in \mathcal{M}, i\in \mathcal{E}.
\end{align}}\vspace{-4mm}

\subsection{Job queues}
After jobs of type $m\in \mathcal{M}$ are admitted to the source of link $i\in \mathcal{E}$, they are injected into a queue $Q_{mi}(t)$ of unsent jobs with queueing law as follows,

\noindent \vspace{-4mm}{\small
\begin{align}\label{eqn:packet-queue}
Q_{mi}(t+1) =& \max\{Q_{mi}(t)-\mu_{mi}(t)-d_{mi}(t)\times s_m,0\}\notag\\
&+r_{mi}(t)\times s_m,~~~~~~~~~~\forall m\in \mathcal{M}, i\in \mathcal{E}.
\end{align}}\vspace{-5mm}

\noindent Here, the length of $Q_{mi}(t)$ is the total number of packets waiting to be delivered at time slot $t$. $d_{mi}(t)$ is the number of type $m$ jobs that are dropped by link $i$ at time slot $t$, as a result of meeting its delay deadline (to be introduced shortly), with

\vspace{-5mm}{\small
\begin{align}\label{eqn:drop-cons}
d_{mi}(t)\in [0, d_{mi}^{max}],\ \forall m\in \mathcal{M}, i\in \mathcal{E}.
\end{align}}\vspace{-5mm}

\noindent where, $d_{mi}^{max}$ is the maximum dropping rate. $\mu_{mi}(t)$ is the number of type $m$ packets delivered over link $i$ at time slot $t$. Since unit-capacity is assumed for each link, we have that

\subsection{Link scheduling and job transmission}
Each link $i\in \mathcal{E}$ is indicated to be either active (transmitting) or idle in each time slot with binary variable $x_i$ as follows,

\vspace{-4mm}{\small
\begin{align}\label{eqn:link-cons}
x_i(t)=\begin{cases}1 & \text{if link $i$ is scheduled in slot $t$}\\ 0 & \text{Otherwise.}\end{cases},\ \forall i\in \mathcal{E}.
\end{align}}\vspace{-4mm}

A feasible link schedule should ensure that no pair of mutually interfering links can be active concurrently, \emph{i.e.},

\vspace{-4mm}{\small
\begin{align}\label{eqn:collision-cons}
x_i(t)+x_j(t)\leq 1,\ \forall j\in \mathcal{C}_i, i\in \mathcal{E}.
\end{align}}\vspace{-5mm}

If link $i$ is active in slot $t$, it needs to decide which type of jobs should be served with the available capacity. Recall that each link has unit capacity, at most one type of jobs can be served in current slot with the following capacity constraint,

\vspace{-4mm}{\small
\begin{align}\label{eqn:capacity-cons}
\sum_{m\in \mathcal{M}}\mu_{mi}(t) = x_i(t),\ \forall i\in \mathcal{E}.
\end{align}}\vspace{-5mm}

%\subsection{Bundled packet transmission}
%
%Unlike the simplified assumption of uniform size for each job and that each packet is independently useful, \emph{i.e.}, the utility gain of one packet delivery does not depend on the correct reception of other packets, it is more practical to consider the case where each job may have non-uniform sizes and the bundled packets in each job contain inter-related contents. Therefore, each job of packets should be either successfully received in a whole or become fully invalid if any one is dropped. That is to say,
%
%\vspace{-4mm}{\small
%\begin{align}\label{eqn:bundle}
%&\text{For each job of type $m$ on link $i$, all its $s_m$ packets should be} \notag \\
%&\text{either fully delivered or completed dropped, }\forall m\in \mathcal{M}, i\in \mathcal{E}.
%\end{align}}\vspace{-4mm}

\vspace{-4mm}{\small
\begin{align}\label{eqn:job-trans-cons}
\mu_{mi}(t)\in \{0,1\},\ \forall m\in \mathcal{M}, i\in \mathcal{E}.
\end{align}}\vspace{-5mm}

\subsection{Worst case delay guarantee}
As stated previously, we novelly address the worst-case delay bound for each admitted job in the network as follows,

\vspace{-4mm}{\small
\begin{align}
&\text{Each type-$m$ job for link $i$ is either scheduled for transmission or}\notag \\ &\text{dropped (subject to a penalty) before its maximum delay $D_m$,}\notag \\ &~~~~~~~~~~~~~~~~~~~~~~~~~~~~~~~~~~~~~~~~~~~~~~~~~~~~~\forall m\in \mathcal{M}, i\in \mathcal{E}.\label{eqn:delay}
\end{align}}\vspace{-5mm}

It is natural that a penalty, $\beta>0$, for each dropped packet should be charged, such that it is not rational for each link $i$ to greedily admit jobs for now while to drop them later.

\subsection{Useful definitions}
We present some important definitions that will be used in the rest of the paper.

\begin{definition}[Queue and Network Stability \cite{book2010}]\label{def:queueStable}
A queue $Q$ is \emph{strongly stable} (or \emph{stable} for short)
 if and only if

\vspace{-3mm}{\small
\begin{equation*}
\lim_{t\rightarrow \infty}\sup \frac{1}{t}\sum_{\tau=0}^{t-1}
\mathbb{E}(Q(\tau))<\infty,
\end{equation*}}\vspace{-4mm}

\noindent where $Q(\tau)$ is the queue size at time slot $\tau$ and
$\mathbb{E}(\cdot)$ is the expectation.
 %\end{definition}
%
%\begin{definition}[Network Stability \cite{book2010}]\label{def:networkStable}
A network is \emph{strongly stable} (or \emph{stable} for short) if
and only if all queues in the network are strongly stable.
\end{definition}

\begin{theorem}[{\small Necessity \& Sufficiency for Queue
Stability \cite{book2010}}]\label{theorem:StabilityNecessity} For any
queue $Q$ with the following queuing law,

\vspace{-4mm}{\small
\begin{align*}
Q(t+1)=Q(t)-\gamma(t)+\alpha(t),
\end{align*}}\vspace{-5mm}

\noindent where $\alpha(t)$ and $\gamma(t)$ are the arrival and departure
rates in time slot $t$, respectively, the following results hold:

\noindent\emph{Necessity}: If queue $Q$ is strongly stable, then its
average incoming rate $\bar{\alpha}=\lim_{t\rightarrow
\infty}\frac{1}{t}\sum_{\tau=0}^{t-1}\mathbb{E}(\alpha(\tau))$ is no
larger than the average outgoing rate $\bar{\gamma}=\lim_{t\rightarrow
\infty}\frac{1}{t}\sum_{\tau=0}^{t-1} \mathbb{E}(\gamma(\tau))$.

\noindent \emph{Sufficiency}: If the average incoming rate $\bar{\alpha}$
is strictly smaller than the average outgoing rate $\bar{\gamma}$,
\emph{i.e.}, $\bar{\alpha}+\epsilon\leq \bar{\gamma}$ with $\epsilon>0$, then
queue $Q$ is strongly stable.
\end{theorem}

Hereinafter, for any variable $\alpha(t)$, we denote its time-averaged value as $\bar{\alpha}$, \emph{i.e.}, $\bar{\alpha}=\lim_{t\rightarrow
\infty}\frac{1}{t}\sum_{\tau=0}^{t-1}\mathbb{E}(\alpha(\tau))$.\vspace{-1mm}

\subsection{Throughput utility maximization problem}
Our objective is to dynamically decide the rate control, link scheduling, and job transmission and dropping, such that the time-averaged net utility (throughput utility minus the job dropping penalty) can be maximized while the worst-case delay is guaranteed for jobs with non-uniform sizes.

\vspace{-4mm}{\small
\begin{align}\label{eqn:opt-problem}
\max&~~\sum_{m\in \mathcal{M}}\sum_{i\in \mathcal{E}}[U(\bar{r}_{mi}\times s_m)-\beta\bar{d}_{mi}\times s_m]\\
s.t.&~~\text{Network stability, and Constraint (\ref{eqn:rate-cons}),(\ref{eqn:drop-cons}),(\ref{eqn:job-trans-cons}),(\ref{eqn:link-cons}),(\ref{eqn:collision-cons}),(\ref{eqn:capacity-cons}),(\ref{eqn:delay})}.\notag
\end{align}}\vspace{-5mm}

\noindent Here, $U(\cdot)$ is the throughput utility function, which is non-negative, non-decreasing, concave and differentiable. It is reasonable to have $\beta > U'(0)$ such that admitting one job into the queue for now while dropping it later brings no positive utility gain.

%Let $\bar{r}_{mi}$ be the average data rate for type-$m$ job over link $i$. $\mathbf{\bar{r}}=<\bar{r}_{11},\ldots,\bar{r}_{mi},\ldots, \bar{r}_{|\mathcal{M}||\mathcal{E}|}>$ is the vector of average throughput for each session. $\mathbf{\bar{r}}\leq \mathbf{\bar{\mu}} + \mathbf{\bar{d}}$ means that there exists some rate control, link scheduling and job-dropping mechanism such that the network is stabilized.

Important notations are summarized in Table \ref{table:notation}.\vspace{-1mm}

\begin{table}[!t]
\vspace{-2mm}
\centering
\begin{tabular}{|p{0.55cm}|p{2.85cm}||p{0.65cm}|p{2.95cm}|}
\hline
  $\mathcal{N}$ & Set of nodes  &  $\mathcal{E}$ & Set of links\\\hline
  $\mathcal{M}$ & Set of job types &  $\mathcal{C}_i$ & Collision set of link $i$ \\\hline
  $\mathbb{E}(\cdot)$ & The expectation & $U(\cdot)$ & Utility function \\\hline
\end{tabular}
\begin{tabular}{|p{1cm}|p{6.95cm}|}
  \hline
  $s_m$ & Size of type-$m$ jobs\\
  \hline
  $s^{max}$ & Maximum job size of all types\\
  \hline
  $D_m$ & Worst-case delay of type-$m$ jobs\\
  \hline
  $A_{mi}(t)$ & Arrival rate of type-$m$ jobs on link $i$ in time slot $t$\\
  \hline
  $A_{mi}^{max}$ & Maximum arrival rate of type-$m$ jobs on link $i$ \\
  \hline
  $r_{mi}(t)$ & Admitted type-$m$ jobs on link $i$ in time slot $t$\\
  \hline
  $\eta_{mi}(t)$ & Auxiliary variable for $r_{mi}(t)$ in time slot $t$\\
  \hline
  $x_{i}(t)$ & Binary var: link $i$ is scheduled in time slot $t$?\\
  \hline
  $\mu_{mi}(t)$ & Binary var: type-$m$ job is transmitted over link $i$ in time slot $t$?\\
  \hline
  $\mu_{mi}(t^-)$ & Binary var: type-$m$ job has not finished transmission over link $i$ in time slot $t$?\\
  \hline
  $d_{mi}(t)$ & $\#$ of dropped type-$m$ jobs on link $i$ in time slot $t$\\
  \hline
  $d_{mi}^{max}(t)$ & Maximum drop rate of type-$m$ jobs on link $i$\\
  \hline
  $Q_{mi}(t)$ & Packet queue of type-$m$ jobs on link $i$ in time slot $t$\\
  \hline
  $Y_{mi}(t)$ & Rate control virtual queue for type-$m$ jobs on link $i$ at time $t$\\
  \hline
  $Z_{mi}(t)$ & Delay virtual queue for type-$m$ jobs on link $i$ at time $t$\\
  \hline
  $\epsilon_{mi}$ & Constant for delay virtual queue $Z_{mi}(t)$\\
  \hline
  $V$ & User-defined positive constant in dynamic algorithm \\
  \hline
  $B$ & Quantity defined in Sec.~\ref{sec:algorithm}\\
  \hline
\end{tabular}\vspace{-1mm}
\caption{List of notations.}\label{table:notation} \vspace{-10mm}
\end{table}

\section{CSMA-based Dynamic Wireless Communication Algorithm}\label{sec:algorithm}%\vspace{-1mm}
In this section, we introduce our CSMA-based wireless communication algorithm, which dynamically decides the rate control, link scheduling, and job transmission and dropping, so as to maximize the time-averaged throughput utility as defined in (\ref{eqn:opt-problem}). We first define two important virtual queues (to deal with rate control and delay bounds, respectively), and then present the algorithm design in details by solving four one-slot optimization problems in each time slot in order to approximate (\ref{eqn:opt-problem}).

%The algorithm runs in a time-slotted fashion. Each time slot is further divided into two consecutive parts: control phase and data phase. The control phase is composed of $W+1$ ($W\geq 2$) mini-slots, while the transmissions are scheduled in the data phase.
%
%In each time slot $t\geq 0$, the status of the network is characterized by a vector of two-tuples: $<<x_1(t),y_1(t)>, \ldots, <x_m(t),y_m(t)>, \ldots, <x_M(t),y_M(t)>>$. Here, $x_m(t)\in \{0,1\}$ ($\forall m\in [1, M]$) denotes whether link $m$ is scheduled in slot $t$:
%\begin{align*}
%x_m(t)=\begin{cases}1 & \text{Link $m$ is scheduled in $t$}\\ 0 & \text{Otherwise.}\end{cases},
%\end{align*}
%and $y_m(t)\in [0, S_m^{(max)}-1]$ ($\forall m\in [1, M]$) indicates the number of packets to be consecutively transmitted in the following time slots.

\subsection{Virtual queues}
We have two types of virtual queues to assist the algorithm design.

\vspace{1mm}\noindent \textbf{Virtual queue for rate control}: To deal with the case when the utility function $U(\cdot)$ is non-linear \cite{book2010}, each link $i\in \mathcal{E}$ has the following virtual queue for its rate control on each job type $m\in \mathcal{M}$,

\vspace{-4mm}{\small
\begin{align}\label{eqn:virtual-queue2}
Y_{mi}(t+1) = \max\{Y_{mi}(t) - r_{mi}(t)\cdot s_m, 0\}& + \eta_{mi}(t)\cdot s_m\notag\\
&\forall m\in \mathcal{M}, i\in \mathcal{E}.
\end{align}}\vspace{-4mm}

\noindent Here, $\eta_{mi}(t)$ is an auxiliary variable with

\vspace{-4mm}{\small
\begin{align}\label{eqn:eta-cons}
\eta_{mi}(t)\in \in [0, A_{mi}^{max}],\ \forall m\in \mathcal{M}, i\in \mathcal{E}.
\end{align}}\vspace{-4mm}

The rationale is that, if virtual queue $Y_{mi}(t)$ is kept stable, we have $\bar{\eta}_{mi}\leq \bar{r}_{mi}$ with Theorem \ref{theorem:StabilityNecessity}, \emph{i.e.}, the time-averaged value of $\eta_{mi}(t)\cdot s_m$ constitutes a lower bound for the average throughput. Later on, we will show that maximizing the utility of $\bar{\eta}_{mi}\cdot s_m$ can approximately maximize the utility of average throughput $\bar{r}_{mi}\cdot s_m$.

\vspace{1mm}\noindent \textbf{Virtual queue for delay bound}: The $\epsilon-$persistence queue \cite{neely-infocom11}\footnote{Note that, in \cite{neely-infocom11}, the $\epsilon-$persistence queue can only handle the case when transmission jobs have the uniform size. In this paper, we adapt this technique to the jobs with non-uniform sizes.} is applied in order to meet the QoS constraint. For each job type $m\in \mathcal{M}$, each link $i\in \mathcal{E}$ maintains the following virtual queue,

\vspace{-4mm}{\small
\begin{align}\label{eqn:virtual-queue}
Z_{mi}(t+1)=&\max\{Z_{mi}(t)+\mathbf{1}_{\{Q_{mi}(t)>0\}}(\epsilon_{mi} -\mu_{mi}(t) )\notag\\
            &-d_{mi}(t)\times s_m-\mathbf{1}_{\{Q_{mi}(t)=0\}}, 0\},~~~~ \forall m\in \mathcal{M}, i\in \mathcal{E}.
\end{align}}\vspace{-4mm}

\noindent Here, $\mathbf{1}_{\{\cdot\}}$ is a binary indicator function. $\epsilon_{mi}$ is a positive constant. The virtual queue $Z_{mi}(t)$ approximately keeps track of the delay information for data packet queue $Q_{mi}(t)$ and assists our algorithm design (to be introduced shortly).

\subsection{Distributed dynamic algorithm}
We derive the dynamic algorithm by decoupling the time-averaged utility maximization problem (\ref{eqn:opt-problem}) into four one-slot optimization problems to be solved in each time slot.

Each link $i\in \mathcal{E}$ maintains a set of queues $\Theta(t)=\{Y_{mi}(t), Q_{mi}(t), Z_{mi}(t)|\forall m\in \mathcal{M}, i\in \mathcal{E}\}$. We define the Lyapunov function as follows,

\vspace{-4mm}{\small
\begin{align}\label{eqn:lyapunov}
L(\Theta(t))=\frac{1}{2}\sum_{m\in \mathcal{M}}\sum_{i\in \mathcal{E}}[(Y_{mi}(t))^2 + (Q_{mi}(t))^2 + (Z_{mi}(t))^2].
\end{align}}\vspace{-4mm}

The one-slot conditional Lyapunov drift is

\vspace{-4mm}{\small
\begin{align}\label{eqn:drift}
\Delta(\Theta(t))=L(\Theta(t+1))-L(\Theta(t)).
\end{align}}\vspace{-4mm}

By squaring the queueing laws in Eqn.~(\ref{eqn:packet-queue}), (\ref{eqn:virtual-queue2}) and (\ref{eqn:virtual-queue}), We can have the \emph{drift-plus-penalty} inequality as follows (derivation details are included in \opt{short}{technical report \cite{hxli-tech-csma}}\opt{long}{Appendix \ref{appendix:derivation}}),

\vspace{-4mm}{\small
\begin{align}\label{eqn:drift-plus-penalty}
&\Delta(\Theta(t))-V\sum_{m\in \mathcal{M}}\sum_{i\in \mathcal{E}}(U(\eta_{mi}(t)\cdot s_m)-\beta d_{mi}(t)\times s_m)\notag\\
\leq & B + \sum_{m\in \mathcal{M}}\sum_{i\in \mathcal{E}}Z_{mi}(t)\cdot \epsilon_{mi} - \Phi_1(t)-\Phi_2(t) - \Phi_3(t)-\Phi_4(t).
\end{align}}\vspace{-4mm}

\noindent Here, $B=\frac{1}{2}\sum_{m\in \mathcal{M}}\sum_{i\in \mathcal{E}}[3(A_{mi}^{max}\cdot s_m)^2 + 2(1+d_{mi}^{max}\cdot s_m)^2 + (\epsilon_{mi})^2]$ is a constant value, and $V>0$ is a user-defined parameter to adjust the weight of net utility in the expression. $\Phi_1(t)$, $\Phi_2(t)$, $\Phi_3(t)$ and $\Phi_4(t)$ are as follows,
\begin{itemize}
\item Terms related to auxiliary variables $\eta_{mi}(t)$:

\vspace{-4mm}{\small
\begin{align*}
\Phi_1(t)=\sum_{m\in \mathcal{M}}\sum_{i\in \mathcal{E}}[V\cdot U(\eta_{mi}(t)s_m)-Y_{mi}(t)\cdot \eta_{mi}(t)s_m].
\end{align*}}\vspace{-4mm}

\item Terms related to rate control variables $r_{mi}(t)$:

\vspace{-4mm}{\small
\begin{align*}
\Phi_2(t)=\sum_{m\in \mathcal{M}}\sum_{i\in \mathcal{E}}r_{mi}(t)\cdot s_m\cdot[Y_{mi}(t)-Q_{mi}(t)].
\end{align*}}\vspace{-4mm}

\item Terms related to link scheduling and job transmission variables $\mu_{mi}(t)$:

\vspace{-4mm}{\small
\begin{align*}
\Phi_3(t)=\sum_{m\in \mathcal{M}}\sum_{i\in \mathcal{E}}\mu_{mi}(t)\cdot[Q_{mi}(t)+Z_{mi}(t)].
\end{align*}}\vspace{-4mm}

\item Terms related to packet drop variables $d_{mi}(t)$:

\vspace{-4mm}{\small
\begin{align*}
\Phi_4(t)=\sum_{m\in \mathcal{M}}\sum_{i\in \mathcal{E}}d_{mi}(t)\cdot s_m\cdot [Q_{mi}(t)+Z_{mi}(t) - V\cdot \beta].
\end{align*}}\vspace{-4mm}
\end{itemize}

According to \emph{Lyapunov optimization theory} \cite{book2010}, we can maximize a lower bound of the time-averaged throughput utility and find optimal solutions to the rate control, link scheduling, job transmission and dropping variables by minimizing the RHS of the \emph{drift-plus-penalty} equality (\ref{eqn:drift-plus-penalty}), observing the queue lengths $\Theta(t)$ and the packet arrival $A_{mi}(t)$ in each time slot $t$. Hence, we propose a dynamic algorithm to solve the one-slot optimization problem in each time slot $t$ as follows,

\vspace{-4mm}{\small
\begin{align}\label{eqn:oneshot-profit}
\max&~~\Phi_1(t)+ \Phi_2(t) + \Phi_3(t) + \Phi_4(t)\\
s.t.&~~\text{Constraints (\ref{eqn:rate-cons}),(\ref{eqn:drop-cons}),(\ref{eqn:job-trans-cons}),(\ref{eqn:link-cons}),(\ref{eqn:collision-cons}),(\ref{eqn:capacity-cons}),(\ref{eqn:eta-cons})}.\notag
\end{align}}\vspace{-4mm}

Note that, the delay constraint (\ref{eqn:delay}) is not included in the one-slot optimization, since it could be satisfied by the stability of virtual queue $Z_{mi}(t)$ to be shown in Sec.~\ref{sec:analysis}.

%\emph{i.e.}, maximizing $\Phi^{(1)}_{i}(t)$, $\Phi^{(2)}_{i}(t)$, $\Phi^{(3)}_{i}(t)$ and $\Phi^{(4)}_{i}(t)$, in each time slot. Hence, we have the algorithm for secondary links as follows.

The maximization problem in (\ref{eqn:oneshot-profit}) can be decoupled into four independent optimization problems:

\vspace{-4mm}{\small
\begin{align}\label{eqn:oneshot-aux}
\max&~~\Phi_1(t)\\
s.t.&~~\text{Constraint (\ref{eqn:eta-cons})},\notag
\end{align}}\vspace{-4mm}

\noindent which is related to the optimal decision on the auxiliary variable $\eta_{mi}(t)$; and

\vspace{-4mm}{\small
\begin{align}\label{eqn:oneshot-rate}
\max&~~\Phi_2(t)\\
s.t.&~~\text{Constraint (\ref{eqn:rate-cons})},\notag
\end{align}}\vspace{-4mm}

\noindent which is related to the optimal decision on the rate control variable $r_{mi}(t)$; and

\vspace{-4mm}{\small
\begin{align}\label{eqn:oneshot-schedule}
\max&~~\Phi_3(t)\\
s.t.&~~\text{Constraint (\ref{eqn:job-trans-cons}),(\ref{eqn:link-cons}),(\ref{eqn:collision-cons}),(\ref{eqn:capacity-cons})},\notag
\end{align}}\vspace{-4mm}

\noindent which is related to the optimal decision on the link scheduling variable $x_{i}(t)$ and job transmission variable $\mu_{mi}(t)$; and

\vspace{-4mm}{\small
\begin{align}\label{eqn:oneshot-drop}
\max&~~\Phi_4(t)\\
s.t.&~~\text{Constraint (\ref{eqn:drop-cons})},\notag
\end{align}}\vspace{-4mm}

\noindent which is related to the optimal decision on the job dropping variable $d_{mi}(t)$. Hence, we have the following dynamic algorithm with optimal solutions to each variable.

\subsubsection{Rate control}
We solve (\ref{eqn:oneshot-aux}) and (\ref{eqn:oneshot-rate}) to decide the auxiliary variables and rate control variables ($\forall m\in \mathcal{M}, i\in \mathcal{E}$) as follows,

\vspace{-4mm}{\small
\begin{align}\label{eqn:rate1}
\eta_{mi}(t) = \max\{\min\{ U'^{-1}(\frac{Y_{mi}(t)}{V})/ s_m, A_{mi}^{max}\}, 0\},
\end{align}}\vspace{-4mm}

\noindent where, $U'^{-1}(\cdot)$ is the reverse function of the first-order derivative of the utility function; and

\vspace{-4mm}{\small
\begin{align}\label{eqn:rate2}
r_{mi} = \begin{cases}A_{mi}(t) & \text{if $Y_{mi}(t)-Q_{mi}(t)>0$}\\ 0 & \text{Otherwise.}\end{cases}
\end{align}}\vspace{-4mm}

\noindent \textbf{Remark}: Virtual queue $Y_{mi}(t)$ can be regarded as the unused tokens for data admission. A large value for $Y_{mi}(t)$ indicates adequate available tokens, which results in fewer new tokens, \emph{i.e.}, $\mu_{mi}(t)$, to be added in this slot. Meanwhile, $Q_{mi}(t)$ reflects the congestion level on the link. $Y_{mi}(t)-Q_{mi}(t)>0$ means we have enough tokens while relatively low congestion. Thus, we admit all the arrived jobs. Otherwise, no job is admitted into the network.

\subsubsection{Link scheduling and job transmission with CSMA}
We design the following mechanism to approximate the optimal solution to (\ref{eqn:oneshot-schedule}). Our CSMA-based scheduling mechanism runs in a two-timescale fashion: super slot and regular slot. Each super slot is composed of $T\geq s^{max}$ ($s^{max}=\max_{m\in \mathcal{M}}\{s_m\}$) regular slots, while each regular slot has the same definition as in \cite{ni-infocom10} and our problem model. The link scheduling decisions are made upon the beginning of each super slot and remain fixed throughout each regular slot in that super slot. However, the served job-types are decided in every regular slot dynamically. To be specific, we have that
\begin{itemize}
\item If $t=nT$ with $n\geq 0$: this is the beginning of the $n^{th}$ super slot. The regular slot of this type is composed of two consecutive phases: control phase and scheduling phase.
    \begin{itemize}
    \item \emph{Control phase}: In this phase, all the links distributively randomly generate a collision-free \emph{control schedule} $\mathbf{z}(t)=[z_1,\ldots,z_i(t),\ldots, z_{E}(t)]$ with $z_i(t)\in \{0,1\},\ \forall i\in \mathcal{E}$. This \emph{control schedule} is not the final decision on link scheduling, but indicates the links which may make changes to its scheduling decision in the \emph{scheduling phase}.

        The \emph{control phase} has $W$ mini-slots\footnote{Compared with the regular slot, the length of mini-slots is negligible.}. At the start of this phase, each link $i$ uniformly randomly select an integer $T_i$ in $[0, W-1]$ and backoff for $T_i$ mini-slots. Link $i$ has the following possible actions:

        $\triangleright$ If link $i$ hears no `INTENT' message before the $(T_i+1)$th mini-slot, it broadcasts an `INTENT' message at mini-slot $T_i+1$.

        \hspace{4mm} -- If there is no collision, link $i$ is included in the \emph{control schedule} and we have $z_i(t)=1$.

        \hspace{4mm} -- Otherwise, link $i$ is not selected into the \emph{control schedule} and we have $z_i(t)=0$.

        $\triangleright$ If link $i$ hears any `INTENT' message before the $(T_i+1)$th mini-slot, it is not included in the \emph{control schedule} and we have $z_i(t)=0$.

    \item \emph{Scheduling phase}: We make the link scheduling decisions based on the \emph{control scheduling} and the link scheduling decisions in previous regular slot as follows:

        $\triangleright$: If $z_i(t)=0$, $x_i(t)=x_i(t-1)$.

        $\triangleright$: If $z_i(t)=1$, we further have that

        \hspace{4mm} -- If there is any active link in link $i$'s collision set, \emph{i.e.}, $\exists j\in \mathcal{C}_i$, $x_j(t-1)=1$, link $i$ is not scheduled in this regular slot and $x_i(t)=0$.

        \hspace{4mm} -- Otherwise, link $i$ randomly becomes active in this regular slot with probability $p_i=\frac{e^{w_i(t)}}{1+e^{w_i(t)}}$, \emph{i.e.},
        \begin{align}
        \begin{cases}x_i(t)=1 & \text{with probability }p_i\\ x_i(t)=0 & \text{with probability }\bar{p}_i=1-p_i.\end{cases}
        \end{align}
        Here, weight $w_{i}(t)=\max_{m\in \mathcal{M}}\{Q_{mi}(t)+Z_{mi}(t)\}$.
    \end{itemize}

\item If $t=nT+\tau$ with $n\geq 0$ and $\tau \in (0, T-1]$: this regular slot is within the $n^{th}$ super slot. We keep the link scheduling decision made in slot $nT$. However, each link can decide which job type is served in this slot:
    \begin{itemize}
    \item If link $i$ is not scheduled in slot $nT$, \emph{i.e.}, $x_i(nT)=0$, it keeps inactive in slot $t$ with $x_i(t)=0$ and $\mu_{mi}(t)=0,\ \forall m\in \mathcal{M}$.

    \item If the transmission job $m^*\in \mathcal{M}$ scheduled in previous slot is not finished, link $i$ goes on with transmitting job type $m^*$ with $x_i(t)=1$, $\mu_{m^*i}(t)=1$ and $\mu_{mi}(t)=0,\ \forall m\in \mathcal{M}, m\neq m^*$. We use $\mu_{mi}(t^-)\in \{0,1\}$ to indicate whether the previously scheduled job is finished or not, with $\mu_{mi}(t^-) = 1$ for unfinished case while $\mu_{mi}(t^-)=0$ for cases where either jobs are completed delivered or no job is scheduled in previous slot.

    \item Otherwise, link $i$ is still active in slot $t$. However, it will select the job type, with maximum weight, to be served in this slot:
    \begin{align}\label{eqn:job-max}
    m^* = arg\max_{m\in \mathcal{M}, (n+1)T-t \geq s_m}\{w_{mi}(t)\}.
    \end{align}
    Here, $(n+1)T-t \geq s_m$ ensures that the selected job can be served before the end of this super slot. Thus, at the beginning of the next super slot, there is no on-going unfinished transmission jobs.
    \end{itemize}
\end{itemize}

\subsubsection{Job drop}
In each time slot $t$, we deterministically decide the number of time-out jobs to be dropped, by solving (\ref{eqn:oneshot-drop}), as follows,

\vspace{-4mm}
{\small
\begin{align}\label{eqn:drop-decision}
d_{mi}(t)=\begin{cases}d_{mi}^{max} & \text{if $Q_{mi}(t)+Z_{mi}(t) > V\beta$}\\ 0 & \text{Otherwise}\end{cases}.
\end{align}}\vspace{-5mm}

\noindent \textbf{Remark}: The rationale is that, each link is reluctant to drop packets until the queue lengths exceed certain threshold, above which we may indicate that packets are suffering a long delay.

The dynamic algorithm is summarized in Alg.~\ref{alg:all}.

%\vspace{-2mm}
{\small
\begin{algorithm}[!h]
\small \caption{Dynamic Net Utility Maximization Algorithm in Time
Slot $t$} \label{alg:all}

\textbf{Input}: $Q_{mi}(t)$, $Y_{mi}(t)$, $Z_{mi}(t)$, $A_{mi}(t)$, $A_{mi}^{max}$, $d_{mi}^{max}$, $(\forall m\in \mathcal{M}, i\in \mathcal{E})$, $V$, $U(\cdot)$ and $\beta$.

\textbf{Output}: $\eta_{mi}(t)$, $r_{mi}(t)$, $x_i(t)$, $\mu_{mi}(t)$, $d_{mi}(t)$, $(\forall m\in [1, M])$.
\begin{algorithmic}[1]
\State \textbf{Rate Control}: For each job-type $m$, link $i$ decides the data admission rate $r_{mi}(t)$ and auxiliary variable
$\eta_{mi}(t)$ by Eqn.~(\ref{eqn:rate2}) and
(\ref{eqn:rate1}), respectively.

\State \textbf{Link Scheduling}: Each link $i$ distributively execute the CSMA algorithm as in Alg.~\ref{alg:CSMA} and find solutions to $x_i(t)$ and $\mu_{mi}(t)$.

\State \textbf{Job Dropping}: For each job-type $m$, link $i$ decides the job dropping rate $d_{mi}(t)$ Eqn.~(\ref{eqn:drop-decision}).

\State Update queues $Q_{mi}(t+1)$, $Y_{mi}(t+1)$ and $Z_{mi}(t+1)$ based
on queuing law (\ref{eqn:packet-queue}), (\ref{eqn:virtual-queue2}) and (\ref{eqn:virtual-queue}), respectively.

\end{algorithmic}
\end{algorithm}
}%\vspace{-4mm}

\vspace{-1mm}
{\small
\begin{algorithm}[!h]
\small \caption{CSMA Scheduling Algorithm at Link $i$ in Time
Slot $t$} \label{alg:CSMA}

\textbf{Input}: $Q_{mi}(t)$, $Z_{mi}(t)$ and $\mu_{mi}(t^-)$, $(\forall m\in \mathcal{M})$.

\textbf{Output}: $x_i(t)$ and $\mu_{mi}(t)$, $(\forall m\in \mathcal{M})$.

\vspace{2mm}
\textbf{If} $(t \mod T)=0$:
\begin{algorithmic}[1]

%\textbf{Control phase}:

\State Uniformly randomly choose an integer $T_i$ from $[1, W]$, and wait for $T_i$ mini-slots;

\textbf{If} link $i$ hears an `INTENT' message from any link in $\mathcal{C}_i$ before the $(T_i+1)$th mini-slot, link $i$ is not included in $\mathbf{z}(t)$ and $z_i(t):=0$. No `INTENT' message will be sent by $i$;

\textbf{Else} Link $i$ broadcasts an `INTENT' message to all links in $\mathcal{C}_i$ at the beginning of the $(T_i+1)$th mini-slot;

\hspace{2mm} \textbf{If} there is a collision, link $i$ is not included in $\mathbf{z}(t)$. Set $z_i(t):=0$.

\hspace{2mm} \textbf{Else}, link $i$ is included in $\mathbf{z}(t)$ by setting $z_i(t):=1$.

%\textbf{Data transmission phase}:

\State \textbf{If} $z_i(t)=0$, set $x_i(t):=x_i(t-1)$;

\State \textbf{Else},

\textbf{If} no link in $\mathcal{C}_i$ was active in slot $t-1$

\hspace{2mm} Set $x_i(t):=1$ with probability $p_i=\frac{e^{w_i(t)}}{1+e^{w_i(t)}}$;

\hspace{2mm} Or, set $x_i(t):=0$ with probability $\bar{p}_i=1-p_i$.

\textbf{Else}, set $x_i(t):=0$.

\State \textbf{If} $x_i(t)=1$, set $\mu_{m^*i}(t):=1$ with $m^*=\max_{m\in \mathcal{M}}\{Q_{mi}(t)+Z_{mi}(t)\}$ and $\mu_{m'i}(t):=0$ with $m'\neq m^*$.

\State \textbf{Else} set $\mu_{m'i}(t):=0$, $\forall m\in \mathcal{M}$.

\end{algorithmic}

\vspace{2mm}
\textbf{If} $(t \mod T)\neq 0$:
\begin{algorithmic}[1]

\State $x_i(t):= x_i(t-1)$.

\State \textbf{If} $x_i(t)=1$

\textbf{If} $\exists m'\in \mathcal{M}$ with $\mu_{m'i}(t^-)=1$, set $\mu_{m'i}(t):=1$ and $\mu_{m''i}(t):=0$ with $m''\neq m'$.

\textbf{Else} set $\mu_{m^*i}(t):=1$ with $m^*=\max_{m\in \mathcal{M}}\{Q_{mi}(t)+Z_{mi}(t)\}$ and $\mu_{m''i}(t):=0$ with $m''\neq m^*$.

\State \textbf{Else} set $\mu_{m'i}(t):=0$, $\forall m\in \mathcal{M}$.

\end{algorithmic}
\end{algorithm}
}%\vspace{-5mm}

\subsection{Computation/communication complexities}
In each time slot, Algorithm \ref{alg:all} incurs polynomial computation and communication complexities at each link $i \in \mathcal{E}$.

\vspace{1mm}
\noindent \textbf{Computation complexity}: For each job type $m\in \mathcal{M}$, link $i\in \mathcal{E}$ finds the optimal solutions to its auxiliary variable, rate control and job dropping variables in constant time with Algorithm \ref{alg:all}. Thus, the overall computation complexity for these variables is in $O(|\mathcal{M}|)$ for each link in each time slot.

For link scheduling and job transmission, each link consumes constant time on the control schedule, at most $O(|\mathcal{E}|)$ complexity to check out the scheduling status of mutual interfering links in previous slot, constant time to compute the link scheduling decision, and $O(|\mathcal{M}|)$ complexity to find the job-type with maximum weight. Hence, the overall complexity for this part is in $O(|\mathcal{E}| + |\mathcal{M}|)$.

To sum up, the computation complexity at each link is in $O(|\mathcal{E}| + |\mathcal{M}|)$.

%We analyze the computation complexity of each component in Algorithm \ref{alg:all} at link $i\in \mathcal{E}$: 1) \emph{rate control}: the auxiliary and rate control variables are computed in constant time with Eqn.~(\ref{eqn:rate1}) and (\ref{eqn:rate2}); 2) \emph{link scheduling and job transmission}:  ; 3) \emph{job dropping}: the number of dropped jobs is decision also with constant time with Eqn.~(\ref{eqn:drop-decision}). Summing up the above, we have that the overall computation complexity of Algorithm \ref{alg:all} is a constant value, independent of the network size, in each time slot.

\vspace{1mm}
\noindent \textbf{Communication complexity}: The only communication overhead occurs at the first step of link scheduling at the beginning of each super slot with Algorithm \ref{alg:CSMA}. If link $i$ timeouts before any of its mutual-interfering links in the control phase, it will just broadcasts one `INTENT' message to its neighborhood; otherwise, no message will be sent by link $i$. If link $i$ is included in the control schedule, it takes at most $O(|\mathcal{E}|)$ communication overhead to find the link scheduling status of its interfering links in previous slot. Therefore, the overall communication complexity for each link is $O(|\mathcal{E}|)$ in each slot.

\section{Performance Analysis}\label{sec:analysis}%\vspace{-1mm}
In this section, we present the analytical results of Algorithm \ref{alg:all}.

\vspace{1mm}
\begin{lemma}[Bounded queue lengths]\label{lemma:bounded-queue}
\emph{Let $Y_{mi}^{max}=V\cdot U'(0)+A_{mi}^{max}\cdot s_m$, $Q_m^{max}=V\cdot U'(0)+2A_{mi}^{max}\cdot s_m$ and $Z_m^{max}=V\cdot \beta/s_m + \epsilon_{mi}$. If $d_{mi}^{max}\geq \max\{A_{mi}^{max}, \epsilon_{mi}/s_m\}$, in each time slot $t\geq 0$, the lengths of packet queues and virtual queues are bounded as follows,}

\vspace{-4mm}{\small
\begin{align}
Q_{mi}(t)\leq Q_{mi}^{max},\ Y_{mi}(t)\leq Y_{mi}^{max},\ \text{and }Z_{mi}(t)\leq Z_m^{max}.
\end{align}}\vspace{-5mm}

\end{lemma}

\vspace{1mm}
This lemma is the basis to prove the worst-case delay guarantee in Theorem \ref{theorem:delay}, and can be proved by induction. Details can be found in \opt{short}{technical report \cite{hxli-tech-csma}}\opt{long}{Appendix \ref{appendix:bounded-queue}}.

\vspace{1mm}
\begin{theorem}[Worst-case delay guarantee]\label{theorem:delay}
\emph{Each job of type $m\in \mathcal{M}$ on link $i\in \mathcal{E}$ is either schedule for transmission or dropped before a preset deadline $D_{mi}$ if we set $\epsilon_{mi}= \frac{Q_{mi}^{max}+Z_{mi}^{max}}{D_{mi}}$.}
\end{theorem}

\vspace{1mm}
We prove this theorem by contradiction with details in \opt{short}{technical report \cite{hxli-tech-csma}}\opt{long}{Appendix \ref{appendix:delay}}.

%We have that, if we have service deadline $D_{mi}$, we should set a value to $\epsilon_{mi}$ such that $\epsilon_{mi}< \frac{Q_{mi}^{max}+Z_{mi}^{max}}{D_{mi}-s_m}$.

\vspace{1mm}
\begin{theorem}[$1-\delta$ weight]\label{theorem:delta-weight}
\emph{Given any $\theta$ and $\delta$ with $0<\theta, \delta< 1$, if $\max\{\Phi_3(t)\}\geq \frac{1}{\theta}(|\mathcal{E}|\log 2 + \log \frac{1}{\delta})$, we have that in any beginning timeslot $nT$ of super frame $n\geq 0$, with probability greater than $1-\delta$, Algorithm \ref{alg:CSMA} finds a schedule $\mu_{mi}(t)$ such that}

\vspace{-4mm}{\small
\begin{align}
\Phi_3(t) \geq (1-\theta)\max\{\Phi_3(t)\}.
\end{align}}\vspace{-4mm}

\end{theorem}

\vspace{1mm}
This theorem is proved by glauber dynamics and time separation assumption (commonly assumed in \cite{ni-infocom10} and references therein, and justified by \cite{jiang-tech08, liu-tech08}). Details are included in \opt{short}{technical report \cite{hxli-tech-csma}}\opt{long}{Appendix \ref{appendix:delta-weight}}. This theorem will be utilized for the proof to Theorem \ref{theorem:optimality}.

%Compared with offline optimum of the throughput optimization problem with deadline constraint while without the consecutive transmission constraint, Algorithm \ref{alg:all} can achieve $\frac{T-S_{max}}{(1+\epsilon)T}$ faction of the capacity region.

%How about the throughput utility?

\vspace{1mm}
\begin{theorem}[Utility-optimality]\label{theorem:optimality}
\emph{The average throughput utility achieved with our proposed Algorithm \ref{alg:all}, $\Psi$, is within a constant gap $\frac{(T-1)}{2V} B''+B' \frac{(T-s^{max}+1)(T-s^{max})}{2VT} + B/V$ from a $(1-\nu)$ offline optimum, $\Psi^{*\nu}$, which uses $(1-\nu)$ fraction of the full capacity region and has perfect information into the future, as follows,}

\vspace{-4mm}{\small
\begin{align}
\Psi\geq \Psi^{*\nu}-\frac{(T-1)}{2V} B''-B' \frac{(T-s^{max}+1)(T-s^{max})}{2VT} - B/V.
\end{align}}\vspace{-4mm}

\noindent \emph{Here, $B=\frac{1}{2}\sum_{m\in \mathcal{M}}\sum_{i\in \mathcal{E}}[3(A_{mi}^{max}s_m)^2 + 2(1+d_{mi}^{max})^2 + (\epsilon_{mi})^2]$, $B'=\sum_{i\in \mathcal{E}}\sum_{m\in \mathcal{M}} ((A_{mi}^{max} + 2d_{mi}^{max})\times s_m + 2)$, $B''=\sum_{m\in \mathcal{M}}\sum_{i\in \mathcal{E}}( (A_{mi}^{max} s_m)^2+ \epsilon_{mi}^2+2d_{mi}^{max}s_m(1+d_{mi}^{max}s_m))$, and $\nu=1-\frac{(1-\delta)(1-\theta)(T-s^{max}+1)}{T} $.}

\end{theorem}

\vspace{1mm}
%This theorem is proved based on Lyapunov optimization and Theorem \ref{theorem:delta-weight}. We include the details in \opt{short}{technical report \cite{hxli-tech-csma}}\opt{long}{Appendix \ref{appendix:optimality}}.

%\opt{long}{
%\section{Proof to Theorem }\label{appendix:optimality}

\begin{proof}

For each time slot $t=nT+\tau$ with $n\geq 0$ and $\tau \in [0, T-s^{max}]$, we have that

\vspace{-4mm}{\small
\begin{align*}
&\sum_{m\in \mathcal{M}}(Q_{mi}(t)+Z_{mi}(t))\mu_{mi}(t-1)\\
=&\sum_{m\in \mathcal{M}}(Q_{mi}(t)+Z_{mi}(t))\mu_{mi}(t^-)\\
 & + \sum_{m\in \mathcal{M}}(Q_{mi}(t)+Z_{mi}(t))[\mu_{mi}(t-1)-\mu_{mi}(t^-)]\\
\leq& \sum_{m\in \mathcal{M}}(Q_{mi}(t)+Z_{mi}(t))\mu_{mi}(t^-)\\
& +\sum_{m\in \mathcal{M}}(Q_{mi}(t)+Z_{mi}(t))(\mu_{mi}(t)-\mu_{mi}(t^-))\\
=&\sum_{m\in \mathcal{M}}(Q_{mi}(t)+Z_{mi}(t))\mu_{mi}(t).
\end{align*}}\vspace{-4mm}

The inequality is based on the fact that, in Algorithm \ref{alg:CSMA}, $(\mu_{mi}(t)-\mu_{mi}(t^-))$ is determined by serving the job-types with maximum queue lengths on each link.

Based on the queueing law Eqn.~(\ref{eqn:packet-queue}) and (\ref{eqn:virtual-queue}), we have that

\vspace{-4mm}{\small
\begin{align*}
&|Q_{mi}(t)+Z_{mi}(t)-Q_{mi}(t-1)-Z_{mi}(t-1)|\\
=&|r_{mi}(t)\cdot s_m - \mu_{mi}(t) - d_{mi}(t)\cdot s_m \\
&+ \mathbf{1}_{\{Q_{mi}(t)>0\}}(\epsilon_{mi} -\mu_{mi}(t) )-d_{mi}(t)\times s_m-\mathbf{1}_{\{Q_{mi}(t)=0\}}|\\
\leq& (A_{mi}^{max} + 2d_{mi}^{max})\times s_m + \epsilon_{mi} + 2.
\end{align*}}\vspace{-4mm}

Thus, we further have that, $\forall i\in \mathcal{E}$,

\vspace{-4mm}{\small
\begin{align}\label{eqn:queue-ineqn}
&\sum_{m\in \mathcal{M}}(Q_{mi}(t-1)+Z_{mi}(t-1))\mu_{mi}(t-1)\notag\\
\leq& B_i' + \sum_{m \in \mathcal{M}}(Q_{mi}(t)+Z_{mi}(t))\mu_{mi}(t),
\end{align}}\vspace{-4mm}

\noindent where, $B_i' = \sum_{m\in \mathcal{M}} ((A_{mi}^{max} + 2d_{mi}^{max})\times s_m + \epsilon_{mi} + 2)$.

For time slot $t=nT+\tau$ with $\tau\in [0, T-s^{max}]$, we have the following based on Eqn.~(\ref{eqn:queue-ineqn}),

\vspace{-4mm}{\small
\begin{align}\label{eqn:queue-ineqn-nT}
&\sum_{m\in \mathcal{M}}\sum_{i\in \mathcal{E}}(Q_{mi}(nT)+Z_{mi}(nT))\mu_{mi}(nT)\notag \\
\leq& \tau B' + \sum_{m \in \mathcal{M}}\sum_{i\in \mathcal{E}}(Q_{mi}(t)+Z_{mi}(t))\mu_{mi}(t).
\end{align}}\vspace{-4mm}

\noindent Here, $B'=\sum_{i\in \mathcal{E}}B_i'$. The inequality is equivalent to the following

\vspace{-4mm}{\small
\begin{align}\label{eqn:queue-ineqn-nT2}
&\Phi_3(nT)\leq \tau B' + \Phi_3(nT+\tau).
\end{align}}\vspace{-4mm}

Summing up this inequality over time slots with $\tau\in [0, T-s^{max}]$, we have that

\vspace{-4mm}{\small
\begin{align}\label{eqn:queue-ineqn-nT-sum}
&(T-s^{max}+1)\sum_{m\in \mathcal{M}}\sum_{i\in \mathcal{E}}(Q_{mi}(nT)+Z_{mi}(nT))\mu_{mi}(nT)\notag\\
\leq & \frac{(T-s^{max}+1)(T-s^{max})}{2} B'\notag\\
&+ \sum_{\tau \in [0, T-s^{max}]}\sum_{m \in \mathcal{M}}\sum_{i\in \mathcal{E}}(Q_{mi}(nT+\tau)\notag\\
&+Z_{mi}(nT+\tau))\mu_{mi}(nT+\tau),
\end{align}}\vspace{-4mm}

\noindent which is equivalent to

\noindent\vspace{-4mm}{\small
\begin{align}\label{eqn:queue-ineqn-nT-sum2}
&(T-s^{max}+1)\Phi_3(nT)\notag\\
\leq & \frac{(T-s^{max}+1)(T-s^{max})}{2} B'+ \sum_{\tau \in [0, T-s^{max}]}\Phi_3(nT+\tau),
\end{align}}\vspace{-4mm}

Based on Theorem \ref{theorem:delta-weight}, we have that, with probability no larger than $\delta$, our link scheduling decisions will result in

\vspace{-4mm}{\small
\begin{align}
0\leq \Phi_3(t) < (1-\theta)\max\{\Phi_3(t)\},\ \forall t=nT, n\geq 0.
\end{align}}\vspace{-4mm}

Thus, taking expectations on $\Phi_3(t)$, we have that

\vspace{-4mm}{\small
\begin{align}
\mathbb{E}(\Phi_3(t)) \geq&  (1-\delta)(1-\theta)\mathbb{E}(\max\{\Phi_3(t)\}) + \delta \cdot 0\notag \\
 =&  (1-\delta)(1-\theta)\mathbb{E}(\max\{\Phi_3(t)\}),\ \forall t=nT, n\geq 0.
\end{align}}\vspace{-4mm}

The above inequality is under the condition that $t=nT$. Next, we study the case when $t=nT+\tau$ with $\tau\in [0, T-s^{max}]$.

When $\tau\in [0, T-s^{max}]$, we can have the following based on Eqn.~(\ref{eqn:queue-ineqn-nT2})

\vspace{-4mm}{\small
\begin{align}
\mathbb{E}(\Phi_3(nT+\tau)) \geq&  -\tau B' + (1-\delta)(1-\theta)\mathbb{E}(\max\{\Phi_3(t)\}),\notag\\
&~~~~~~~~~~ \forall t=nT, n\geq 0.
\end{align}}\vspace{-4mm}

Since the job arrival is i.i.d., we know that, for any fraction $1-\nu$ of the full capacity region, there exists a stationary randomized algorithm solving the rate control, link scheduling and job dropping decisions with offline optimal throughput utility \cite{book2010}. We denote the optimal solutions, with this stationary randomized algorithm in $1-\nu$ capacity region, as $\eta_{mi}^{*\nu}(t)$, $r_{mi}^{*\nu}(t)$, $x_i^{*\nu}(t)$, $\mu_{mi}^{*\nu}(t)$ and $d_{mi}^{*\nu}(t)$, respectively. Let $\Phi_1^{*\nu}(t)$, $\Phi_2^{*\nu}(t)$, $\Phi_3^{*\nu}(t)$ and $\Phi_4^{*\nu}(t)$ denote the value of these four expressions under the stationary randomized algorithm. Denote $\eta_{mi}^{*\nu}=\mathbb{E}\{\eta_{mi}^{*\nu}(t)\}$, $r_{mi}^{*\nu}=\mathbb{E}\{r_{mi}^{*\nu}(t)\}$, $x_{i}^{*\nu}=\mathbb{E}\{x_i^{*\nu}(t)\}$, $\mu_{mi}^{*\nu}=\mathbb{E}\{\mu_{mi}^{*\nu}(t)\}$ and $d_{mi}^{*\nu}=\mathbb{E}\{d_{mi}^{*\nu}(t)\}$.

In Sec.~\ref{sec:algorithm}, we have seen that our solutions to auxiliary variables, rate control and job-drop decisions maximize the value of $\Phi_1(t)$, $\Phi_2(t)$ and $\Phi_4$ in each time slot, \emph{i.e.}, $\Phi_1(t)\geq \Phi_1^{*\nu}(t)$, $\Phi_2(t)\geq \Phi_2^{*\nu}(t)$ and $\Phi_4(t)\geq \Phi_4^{*\nu}(t)$. Then, for time $t=nT+\tau$ with $\tau\in [0, T-s^{max}]$, we have that

\vspace{-4mm}{\small
\begin{align*}
&B + \sum_{m\in \mathcal{M}}\sum_{i\in \mathcal{E}} \mathbb{E}(Z_{mi}(t)) \epsilon_{mi} -\mathbb{E}(\Phi_1(t))-\mathbb{E}(\Phi_2(t))\\
&-\mathbb{E}(\Phi_3(t))-\mathbb{E}(\Phi_4(t))\\
\leq& B + \sum_{m\in \mathcal{M}}\sum_{i\in \mathcal{E}} \mathbb{E}(Z_{mi}(t))\cdot \epsilon_{mi}  - \mathbb{E}(\Phi_1^{*\nu}(t)) - \mathbb{E}(\Phi_2^{*\nu}(t)) + \tau B' \\
&- (1-\delta)(1-\theta)\mathbb{E}(\Phi_3^{*0}(nT)) - \mathbb{E}(\Phi_4^{*\nu}(t))
\end{align*}}\vspace{-4mm}

For time $t=nT+\tau$ with $\tau\in (T-s^{max}, T)$, we have that

\vspace{-4mm}{\small
\begin{align*}
&B + \sum_{m\in \mathcal{M}}\sum_{i\in \mathcal{E}} \mathbb{E}(Z_{mi}(t))\cdot \epsilon_{mi} -\mathbb{E}(\Phi_1(t))-\mathbb{E}(\Phi_2(t))\\
&-\mathbb{E}(\Phi_3(t))-\mathbb{E}(\Phi_4(t))\\
\leq& B + \sum_{m\in \mathcal{M}}\sum_{i\in \mathcal{E}} \mathbb{E}(Z_{mi}(t))\cdot \epsilon_{mi}  - \mathbb{E}(\Phi_1^{*\nu}(t)) - \mathbb{E}(\Phi_2^{*\nu}(t)) - \mathbb{E}(\Phi_4^{*\nu}(t)),
\end{align*}}\vspace{-4mm}

\noindent which is based on the fact that $\Phi_3(t)\geq 0$.

Recalling the \emph{drift-plus-penalty} inequality (\ref{eqn:drift-plus-penalty}), we have that, for time $t=nT+\tau$ with $\tau\in [0, T-s^{max}]$,

\vspace{-4mm}{\small
\begin{align*}
&V\mathbb{E}(\sum_{m\in \mathcal{M}}\sum_{i\in \mathcal{E}}(U(\eta_{mi}(t)\cdot s_m)-\beta d_{mi}(t)\cdot s_m))\\
\geq& \mathbb{E}(\Phi_1^{*\nu}(t)) + \mathbb{E}(\Phi_2^{*\nu}(t)) - \tau B' + (1-\delta)(1-\theta)\mathbb{E}(\Phi_3^{*0}(nT))\\
 &+ \mathbb{E}(\Phi_4^{*\nu}(t))-B -\sum_{m\in \mathcal{M}}\sum_{i\in \mathcal{E}} \mathbb{E}(Z_{mi}(t))\cdot \epsilon_{mi},
\end{align*}}\vspace{-4mm}

\noindent for time $t=nT+\tau$ with $\tau\in (T-s^{max}, T)$,

\vspace{-4mm}{\small
\begin{align*}
&V\mathbb{E}(\sum_{m\in \mathcal{M}}\sum_{i\in \mathcal{E}}(U(\eta_{mi}(t)\cdot s_m)-\beta d_{mi}(t)\cdot s_m))\\
\geq& \mathbb{E}(\Phi_1^{*\nu}(t)) + \mathbb{E}(\Phi_2^{*\nu}(t)) + \mathbb{E}(\Phi_4^{*\nu}(t))-B\\
& - \sum_{m\in \mathcal{M}}\sum_{i\in \mathcal{E}} \mathbb{E}(Z_{mi}(t))\cdot \epsilon_{mi}.
\end{align*}}\vspace{-4mm}

Summing up the above two inequalities for $t=nT+\tau$ with $\tau\in [0, T)$, we have that

\vspace{-4mm}{\small
\begin{align*}
&V\sum_{\tau \in [0, T)}\mathbb{E}(\sum_{m\in \mathcal{M}}\sum_{i\in \mathcal{E}}(U(\eta_{mi}(nT+\tau)\cdot s_m)\\
&-\beta d_{mi}(nT+\tau)\cdot s_m))\\
\geq& \sum_{\tau \in [0, T)}(\mathbb{E}(\Phi_1^{*\nu}(t))+ \mathbb{E}(\Phi_2^{*\nu}(t))+ \mathbb{E}(\Phi_4^{*\nu}(t)))  \\
&+ (1-\delta)(1-\theta)(T-s^{max}+1)\mathbb{E}(\Phi_3^{*0}(nT))\\
& - B'\sum_{\tau \in [0,T-s^{max}]}\tau  -T B\\
& - \sum_{\tau \in [0, T)} \sum_{m\in \mathcal{M}}\sum_{i\in \mathcal{E}} \mathbb{E}(Z_{mi}(nT+\tau))\cdot \epsilon_{mi}.
\end{align*}}\vspace{-4mm}

We expand the right-hand side of the above inequality and have that

\vspace{-4mm}{\small
\begin{align*}
&\sum_{\tau \in [0, T)}\sum_{m\in \mathcal{M}}\sum_{i\in \mathcal{E}}V\mathbb{E}(U(\eta_{mi}(nT+\tau)\cdot s_m)-\beta d_{mi}(nT+\tau)s_m)\\
\geq & \sum_{\tau \in [0, T)}\sum_{m\in \mathcal{M}}\sum_{i\in \mathcal{E}}V\mathbb{E}(U(\eta_{mi}^{*\nu}(nT+\tau)\cdot s_m)-\beta d_{mi}^{*\nu}(nT+\tau)s_m)\\
&+ \sum_{\tau \in [0, T)}\sum_{m\in \mathcal{M}}\sum_{i\in \mathcal{E}}\mathbb{E}( Y_{mi}(t)\cdot s_m \cdot(r_{mi}^{*\nu}(nT+\tau) - \eta_{mi}^{*\nu}(nT+\tau)) )\\
&- \sum_{\tau \in [0, T)}\sum_{m\in \mathcal{M}}\sum_{i\in \mathcal{E}}\mathbb{E}(r_{mi}^{*\nu}(nT+\tau)\cdot s_m \cdot Q_{mi}(nT+\tau) )\\
&+ (1-\delta)(1-\theta)(T-s^{max}+1)\sum_{m\in \mathcal{M}}\sum_{i\in \mathcal{E}}\mathbb{E}(\mu_{mi}^{*0}(nT)\\
&\times(Q_{mi}(nT)+Z_{mi}(nT)))\\
& + \sum_{\tau \in [0, T)}\sum_{m\in \mathcal{M}}\sum_{i\in \mathcal{E}}\mathbb{E}(d_{mi}^{*\nu}(nT+\tau)s_m(Q_{mi}(nT+\tau)\\
&+ Z_{mi}(nT+\tau)))\\
&-B' \frac{(T-s^{max}+1)(T-s^{max})}{2} - T B\\
& - \sum_{\tau \in [0, T)} \sum_{m\in \mathcal{M}}\sum_{i\in \mathcal{E}} \mathbb{E}(Z_{mi}(nT+\tau))\cdot \epsilon_{mi}\\
\geq & \sum_{\tau \in [0, T)}\sum_{m\in \mathcal{M}}\sum_{i\in \mathcal{E}}V(U(\eta_{mi}^{*\nu}\cdot s_m)-\beta d_{mi}^{*\nu}s_m)\\
&- T\sum_{m\in \mathcal{M}}\sum_{i\in \mathcal{E}}\mathbb{E}(r_{mi}^{*\nu}\cdot s_m \cdot Q_{mi}(nT))\\
&- \sum_{\tau \in [0, T)} \sum_{m\in \mathcal{M}}\sum_{i\in \mathcal{E}}\tau (A_{mi}^{max} s_m)^2\\
& - T\sum_{m\in \mathcal{M}}\sum_{i\in \mathcal{E}}\mathbb{E}(Z_{mi}(nT)\cdot \epsilon_{mi})\\
& - \sum_{\tau \in [0, T)} \sum_{m\in \mathcal{M}}\sum_{i\in \mathcal{E}}\tau \epsilon_{mi}^2\\
&+ (1-\delta)(1-\theta)(T-s^{max}+1)\sum_{m\in \mathcal{M}}\sum_{i\in \mathcal{E}}\mathbb{E}(\mu_{mi}^{*0}\\
&\times(Q_{mi}(nT)+Z_{mi}(nT)))
\end{align*}
\begin{align*}
& + \sum_{\tau \in [0, T)}\sum_{m\in \mathcal{M}}\sum_{i\in \mathcal{E}}\mathbb{E}(d_{mi}^{*\nu}s_m(Q_{mi}(nT)+ Z_{mi}(nT)))\\
& - \sum_{\tau \in [0, T)}\sum_{m\in \mathcal{M}}\sum_{i\in \mathcal{E}}(2d_{mi}^{max}s_m\tau(1+d_{mi}^{max}s_m))\\
&-B' \frac{(T-s^{max}+1)(T-s^{max})}{2} - T B.
\end{align*}}\vspace{-4mm}

The second inequality comes from the facts that:
\begin{itemize}
\item Since the stationary randomized algorithm should stabilize the network, each virtual queue $Y_{mi}(t)$ is also stable. Thus, we have that

    \vspace{-4mm}{\small
    \begin{align*}
    \mathbb{E}(r_{mi}^{*\nu}(nT+\tau) - \eta_{mi}^{*\nu}(nT+\tau)) \geq 0.
    \end{align*}}\vspace{-4mm}

\item According to the queueing law in Eqn.~(\ref{eqn:packet-queue}), we have that $Q_{mi}(nT+\tau)-Q_{mi}(nT)\leq \tau A_{mi}^{max}s_m$. Meanwhile, it is a fact that $r_{mi}^{*\nu}=\mathbb{E}(r_{mi}^{*\nu}(nT+\tau))$ So, we further have that

    \vspace{-4mm}{\small
    \begin{align*}
    &\mathbb{E}(r_{mi}^{*\nu}(nT+\tau)\cdot s_m \cdot Q_{mi}(nT+\tau))\\
     \leq& \mathbb{E}(r_{mi}^{*\nu}\cdot s_m \cdot (Q_{mi}(nT)+ \tau A_{mi}^{max} s_m))\\
      \leq& \mathbb{E}(r_{mi}^{*\nu}\cdot s_m \cdot Q_{mi}(nT)+ \tau (A_{mi}^{max} s_m)^2).
    \end{align*}}\vspace{-4mm}

\item According to the queueing law in Eqn.~(\ref{eqn:virtual-queue}), we have that $Z_{mi}(nT+\tau)-Z_{mi}(nT)\leq \tau \epsilon_{mi}$. So, we further have that

    \vspace{-4mm}{\small
    \begin{align*}
    &\mathbb{E}(\epsilon_{mi} \cdot Z_{mi}(nT+\tau))\\
     \leq& \mathbb{E}(\epsilon_{mi}\cdot (Z_{mi}(nT)+ \tau \epsilon_{mi}))\\
      \leq& \mathbb{E}(\epsilon_{mi} \cdot Z_{mi}(nT)+ \tau (\epsilon_{mi})^2).
    \end{align*}}\vspace{-4mm}

\item According to the queueing laws in Eqn.~(\ref{eqn:packet-queue}) and (\ref{eqn:virtual-queue}), we have that $Q_{mi}(nT+\tau)-Q_{mi}(nT)\geq \tau (1+d_{mi}^{max}s_m)$ and $Z_{mi}(nT+\tau)-Z_{mi}(nT)\geq \tau (1+d_{mi}^{max}s_m)$. Based on its definition, we also know that $d_{mi}^{*\nu}=\mathbb{E}(d_{mi}^{*\nu}(nT+\tau))$. Thus, we have that

    \vspace{-4mm}{\small
    \begin{align*}
    &\mathbb{E}(d_{mi}^{*\nu}(nT+\tau)\cdot s_m \cdot (Q_{mi}(nT+\tau)+Z_{mi}(nT+\tau)))\\
    \geq& \mathbb{E}(d_{mi}^{*\nu}\cdot s_m \cdot (Q_{mi}(nT)+ Z_{mi}(nT)- 2\tau(1+ d_{mi}^{max} s_m)))\\
    \geq& \mathbb{E}(d_{mi}^{*\nu}\cdot s_m \cdot (Q_{mi}(nT)+ Z_{mi}(nT))\\
    &- 2d_{mi}^{max}\cdot s_m\tau(1+ d_{mi}^{max} s_m)).
    \end{align*}}\vspace{-4mm}
\end{itemize}

Let $1-\nu = \frac{(1-\delta)(1-\theta)(T-s^{max}+1)}{T}$. We have that

\vspace{-4mm}{\small
\begin{align*}
\frac{(1-\delta)(1-\theta)(T-s^{max}+1)}{T}\mu_{mi}^{*0}=\mu_{mi}^{*\nu}.
\end{align*}}\vspace{-4mm}

Since the stationary randomized algorithm stabilized the network, including all packet queues, we know that

\vspace{-4mm}{\small
\begin{align*}
%&\eta_{mi}^{*\nu}\leq r_{mi}^{*\nu},\\
&r_{mi}^{*\nu}s_m\leq \mu_{mi}^{*\nu} + d_{mi}^{*\nu}s_m,\\
&\epsilon_{mi}\leq \mu_{mi}^{*\nu} + d_{mi}^{*\nu}s_m.
\end{align*}}\vspace{-4mm}

Then, we can have that

\vspace{-4mm}{\small
\begin{align*}
&\sum_{\tau \in [0, T)}\sum_{m\in \mathcal{M}}\sum_{i\in \mathcal{E}}V\mathbb{E}(U(\eta_{mi}(nT+\tau)\cdot s_m)-\beta d_{mi}(nT+\tau)s_m)\\
\geq & \sum_{\tau \in [0, T)}\sum_{m\in \mathcal{M}}\sum_{i\in \mathcal{E}}V(U(\eta_{mi}^{*\nu}\cdot s_m)-\beta d_{mi}^{*\nu}s_m)\\
&- \frac{T(T-1)}{2} \sum_{m\in \mathcal{M}}\sum_{i\in \mathcal{E}} [(A_{mi}^{max} s_m)^2 + \epsilon_{mi}^2]\\
& - \frac{T(T-1)}{2}\sum_{m\in \mathcal{M}}\sum_{i\in \mathcal{E}}(2d_{mi}^{max}s_m(1+d_{mi}^{max}s_m))\\
&-B' \frac{(T-s^{max}+1)(T-s^{max})}{2} - T B.
\end{align*}}\vspace{-4mm}

On both sides, we sum up over all $n\geq 0$, divide by $VnT$ and take limits on $n\rightarrow \infty$\footnote{It should be noted that $\bar{a}=\lim_{t\rightarrow \infty}\frac{1}{t}\sum_{\tau\in [1,t]}\mathbb{E}(a(\tau))$}. We can get that

\vspace{-4mm}{\small
\begin{align*}
&\sum_{m\in \mathcal{M}}\sum_{i\in \mathcal{E}}(U(\bar{\eta}_{mi}\cdot s_m)-\beta \bar{d}_{mi}s_m)\\
\geq & \sum_{m\in \mathcal{M}}\sum_{i\in \mathcal{E}}(U(\eta_{mi}^{*\nu}\cdot s_m)-\beta d_{mi}^{*\nu}s_m)\\
&- \frac{(T-1)}{2V} B''-B' \frac{(T-s^{max}+1)(T-s^{max})}{2VT} - B/V,
\end{align*}}\vspace{-4mm}

\noindent with $B''=\sum_{m\in \mathcal{M}}\sum_{i\in \mathcal{E}}( (A_{mi}^{max} s_m)^2+ \epsilon_{mi}^2+2d_{mi}^{max}s_m(1+d_{mi}^{max}s_m))$.

With Algorithm \ref{alg:all}, we know that $\bar{r}_{mi}\geq \bar{\eta}_{mi}$. Meanwhile, the stationary randomized algorithm can make $r^{*\nu}=\eta^{*\nu}$. Then, we finally have that

\vspace{-4mm}{\small
\begin{align*}
&\sum_{m\in \mathcal{M}}\sum_{i\in \mathcal{E}}(U(\bar{r}_{mi}\cdot s_m)-\beta \bar{d}_{mi}s_m)\\
\geq & \sum_{m\in \mathcal{M}}\sum_{i\in \mathcal{E}}(U(r_{mi}^{*\nu}\cdot s_m)-\beta d_{mi}^{*\nu}s_m)\\
&- \frac{(T-1)}{2V} B''-B' \frac{(T-s^{max}+1)(T-s^{max})}{2VT} - B/V,
\end{align*}}\vspace{-4mm}

\end{proof}

\vspace{1mm}\noindent \textbf{Remark} (\emph{throughput-delay tradeoff}): If we let $V\rightarrow \infty$, the queue lengths will grow to infinitely large (Lemma \ref{lemma:bounded-queue}). With Lemma \ref{lemma:bounded-queue} and Theorem \ref{theorem:delta-weight}, we see that $(1-\delta)(1-\theta)\rightarrow 1$ in this case. If we further let $T\rightarrow \infty$, we will have $1-\nu = \frac{(1-\delta)(1-\theta)(T-s^{max}+1)}{T}\rightarrow 1$, which means $\sum_{m\in \mathcal{M}}\sum_{i\in \mathcal{E}}(U(r_{mi}^{*\nu}\cdot s_m)-\beta d_{mi}^{*\nu}s_m)$ will be arbitrarily close to the offline optimum within the full capacity region instead of a fraction. In addition, if $T/V\rightarrow 0$, the constant utility gap will become $\frac{(T-1)}{2V} B'' + B' \frac{(T-s^{max}+1)(T-s^{max})}{2VT} + B/V\rightarrow 0$. Nevertheless, with Theorem \ref{theorem:delay}, we have to be able to tolerate a long worst-case delay, which is proportional to $V$.

In conclusion, there is a tradeoff between the utility optimality and the tolerable worst-case delay. If $V\rightarrow \infty$, $T\rightarrow \infty$ and $T/V\rightarrow 0$, we will achieve an utility arbitrarily close to the offline optimum at the cost of infinitely large delay.

%vspace{1mm}\noindent \textbf{Remark} (\emph{maximum allowable delay}): With 

\section{Conclusion and Remarks}\label{sec:conclusion}%\vspace{-1mm}
In this paper, we investigate the optimal design of CSMA-based wireless communication to achieve a time-averaged maximum throughput utility with worst-case delay bounds and non-uniform job sizes, in general network settings. A two-timescale dynamic algorithm is proposed with dynamic decisions on rate control, link scheduling, job transmission and dropping in each time slot. Through rigorous analysis, we demonstrate that the proposed protocol can achieve a throughput utility arbitrarily close to its offline optima for jobs with non-uniform sizes, with a worst-case delay guarantees and a tradeoff of infinitely large maximum allowable delay. As our future work, we will explore the optimal CSMA protocol design in multi-channel settings with provable delay performance for general network topologies.

%This paper addresses the low-complexity throughput-utility optimal data dissemination in wireless networks for communication jobs with worst-case delay bounds and non-uniform sizes. As our future work, we will investigate the. 

\bibliographystyle{abbrv}

\begin{appendices}

\opt{long}{

\section{Derivation of the \emph{Drift-plus-Penalty} inequality (\ref{eqn:drift-plus-penalty})}\label{appendix:derivation}

We get the inequality (\ref{eqn:drift-plus-penalty}) based on the following fact: if $a,b,c \geq 0$, $b\leq b^{max}$ and $c\leq c^{max}$, then we have

\vspace{-4mm}{\small
\begin{align}
&(\max\{a-b,0\}+c)^2\notag\\
=& (\max\{a-b,0\})^2 + c^2 + 2c \cdot \max\{a-b,0\}\notag\\
                    \leq& (a-b)^2 + c^2 + 2c (\mathbf{1}_{\{a-b\geq 0\}}\{a-b\} + \mathbf{1}_{\{a-b<0\}}\{0\})\notag\\
                    =&a^2+ b^2 + c^2 -2ab + 2c(\mathbf{1}_{\{a-b\geq 0\}}\{a-b\} + \mathbf{1}_{\{a-b<0\}}\{a-a\})\notag\\
                    =&a^2+ b^2 + c^2 + 2(\mathbf{1}_{\{a-b\geq 0\}}\{ca-cb-ab\}\notag\\
                    & + \mathbf{1}_{\{a-b<0\}}\{ca-ca-ab\})\notag\\
                    \leq& a^2+ (b^{max})^2 + (c^{max})^2 + 2(\mathbf{1}_{\{a-b\geq 0\}}\{ca-ab\} \notag\\
                    &+ \mathbf{1}_{\{a-b<0\}}\{ca-ab\})\notag\\
                    =&a^2+ (b^{max})^2 + (c^{max})^2 - 2a(b-c).\label{eqn:fact}
\end{align}}\vspace{-4mm}

The derivation details of the \emph{Drift-plus-Penalty} inequality (\ref{eqn:drift-plus-penalty}) are as follows.

\vspace{1mm}\noindent $\triangleright$ We know that $Q_{mi}(t)\geq 0$, $0\leq \mu_{mi}(t)\leq 1$, $0\leq d_{mi}(t)\leq d_{mi}^{max}$ and $0\leq r_{mi}(t)\leq A_{mi}^{max}$. With queueing law (\ref{eqn:packet-queue}) and Eqn.~(\ref{eqn:fact}), we have

\vspace{-4mm}{\small
\begin{align*}
[Q_{mi}(t+1)]^2\leq & [Q_{mi}(t)]^2+ (1+d_{mi}^{max}\cdot s_m)^2 + (A_{mi}^{max}\cdot s_m)^2 \\
&- 2Q_{mi}(t)[\mu_{mi}(t)+d_{mi}(t)\cdot s_m - r_{mi}(t)\cdot s_m],\\
&~~~~~~~~~~\forall m\in \mathcal{M}, i\in \mathcal{E}.
\end{align*}}\vspace{-4mm}

\vspace{1mm}\noindent $\triangleright$ We know that $Y_{mi}(t)\geq 0$, $0\leq r_{mi}(t)\leq A_{mi}^{max}$ and $0\leq \eta_{mi}(t)\leq A_{mi}^{max}$. With queueing law (\ref{eqn:virtual-queue2}) and Eqn.~(\ref{eqn:fact}), we have

\vspace{-4mm}{\small
\begin{align*}
[Y_{mi}(t+1)]^2\leq & [Y_{mi}(t)]^2+ (A_{mi}^{max}s_m)^2 + (A_{mi}^{max}s_m)^2 \\
&- 2Y_{mi}(t)s_m[r_{mi}(t) - \eta_{mi}(t)],\\
&~~~~~~~~~~\forall m\in \mathcal{M}, i\in \mathcal{E}.
\end{align*}}\vspace{-4mm}

\vspace{1mm}\noindent $\triangleright$ With queueing law (\ref{eqn:virtual-queue}), we have

\vspace{-4mm}{\small
\begin{align*}
Z_{mi}(t+1)\leq&\max\{Z_{mi}(t)-d_{mi}(t)\times s_m\notag\\
            &-\mathbf{1}_{\{Q_{mi}(t)>0\}}\mu_{mi}(t)-\mathbf{1}_{\{Q_{mi}(t)=0\}}, 0\}+\epsilon_{mi},\notag\\
&~~~~~~~~~~~~~~~~~~~~ \forall m\in \mathcal{M}, i\in \mathcal{E}.
\end{align*}}\vspace{-4mm}

Since $0\leq \mu_{mi}(t)\leq 1$, we also have

\vspace{-4mm}{\small
\begin{align*}
\mu_{mi}(t)\leq \mathbf{1}_{\{Q_{mi}(t)>0\}}\mu_{mi}(t)+\mathbf{1}_{\{Q_{mi}(t)=0\}} \leq 1.
\end{align*}}\vspace{-4mm}

We know that $Z_{mi}(t)\geq 0$, $0\leq \mu_{mi}(t)\leq 1$ and $0\leq d_{mi}(t)\leq d_{mi}^{max}$. With Eqn.~(\ref{eqn:fact}), we can further have

\vspace{-4mm}{\small
\begin{align*}
[Z_{mi}(t+1)]^2\leq & [Z_{mi}(t)]^2+ (1+d_{mi}^{max}\cdot s_m)^2 + (\epsilon_{mi})^2 \\
&- 2Z_{mi}(t)[\mu_{mi}(t)+d_{mi}(t)\cdot s_m  - \epsilon], \\
%\leq & [Z_{mi}(t)]^2+ (1+d_{mi}^{max}\cdot s_m)^2 + (\epsilon_{mi})^2 \\
%&- 2Z_{mi}(t)[\mu_{mi}(t)+d_{mi}(t)\cdot s_m ],\\
&~~~~~~~~~~\forall m\in \mathcal{M}, i\in \mathcal{E}.
\end{align*}}\vspace{-4mm}

According to the above results and the definitions of Lyapunov function in Eqn.~(\ref{eqn:lyapunov}) and \emph{one-slot conditional Lyapunov drift} in Eqn.~(\ref{eqn:drift}), we have

\vspace{-4mm}{\small
\begin{align*}
&\Delta(\Theta(t))-V\sum_{m\in \mathcal{M}}\sum_{i\in \mathcal{E}}(U(\eta_{mi}(t)\cdot s_m)-\beta d_{mi}(t)\times s_m)\notag\\
\leq & \frac{1}{2}\sum_{m\in \mathcal{M}}\sum_{i\in \mathcal{E}}[(1+d_{mi}^{max}\cdot s_m)^2 + (A_{mi}^{max}\cdot s_m)^2]\notag\\
     & + \frac{1}{2}\sum_{m\in \mathcal{M}}\sum_{i\in \mathcal{E}}[(A_{mi}^{max}s_m)^2 + (A_{mi}^{max}s_m)^2]\notag\\
     & + \frac{1}{2}\sum_{m\in \mathcal{M}}\sum_{i\in \mathcal{E}}[(1+d_{mi}^{max}\cdot s_m)^2 + (\epsilon_{mi})^2]\notag\\
     & - \sum_{m\in \mathcal{M}}\sum_{i\in \mathcal{E}}Q_{mi}(t)[\mu_{mi}(t)+d_{mi}(t)\cdot s_m - r_{mi}(t)\cdot s_m]\notag\\
     & - \sum_{m\in \mathcal{M}}\sum_{i\in \mathcal{E}}Y_{mi}(t)s_m[r_{mi}(t) - \eta_{mi}(t)]\notag\\
     & - \sum_{m\in \mathcal{M}}\sum_{i\in \mathcal{E}}Z_{mi}(t)[\mu_{mi}(t)+d_{mi}(t)\cdot s_m - \epsilon_{mi}]\notag\\
     & - V\sum_{m\in \mathcal{M}}\sum_{i\in \mathcal{E}}(U(\eta_{mi}(t)\cdot s_m)-\beta d_{mi}(t)\times s_m)\notag\\
= & \frac{1}{2}\sum_{m\in \mathcal{M}}\sum_{i\in \mathcal{E}}[3(A_{mi}^{max}s_m)^2+ 2(1+d_{mi}^{max}\cdot s_m)^2 + (\epsilon_{mi})^2]\notag\\
     & - \sum_{m\in \mathcal{M}}\sum_{i\in \mathcal{E}}[V\cdot U(\eta_{mi}(t)\cdot s_m) - Y_{mi}(t)\cdot \eta_{mi}(t)s_m]\notag\\
     & - \sum_{m\in \mathcal{M}}\sum_{i\in \mathcal{E}}r_{mi}(t)\cdot s_m \cdot[Y_{mi}(t)-Q_{mi}(t)]\notag\\
     & - \sum_{m\in \mathcal{M}}\sum_{i\in \mathcal{E}}\mu_{mi}(t)\cdot [Q_{mi}(t) + Z_{mi}(t)]\notag\\
     & - \sum_{m\in \mathcal{M}}\sum_{i\in \mathcal{E}}d_{mi}(t)\cdot s_m\cdot [Q_{mi}(t)+Z_{mi}(t)-V\cdot \beta]\notag\\
     & + \sum_{m\in \mathcal{M}}\sum_{i\in \mathcal{E}} Z_{mi}(t)\cdot \epsilon_{mi}\\
= & B + \sum_{m\in \mathcal{M}}\sum_{i\in \mathcal{E}} Z_{mi}(t)\cdot \epsilon_{mi} - \Phi_1(t)-\Phi_2(t) - \Phi_3(t)-\Phi_4(t).
\end{align*}}\vspace{-4mm}

}

\opt{long}{
\section{Proof to Lemma \ref{lemma:bounded-queue}}\label{appendix:bounded-queue}

\begin{proof}
This lemma can be proved by induction based on the rate control and job-dropping decisions in Algorithm \ref{alg:all}.

\textbf{Induction basis}: At time slot $0$, every packet queue and virtual queue are empty. So we have that

\vspace{-4mm}{\small
\begin{align*}
Q_{mi}(0)=0\leq Q_{mi}^{max},\ Y_{mi}(0)=0\leq Y_{mi}^{max},\ Z_{mi}(0)=0\leq Z_{mi}^{max}.
\end{align*}}\vspace{-4mm}

\textbf{Induction steps}: For time slot $t>0$, let $Q_{mi}(t-1)\leq Q_{mi}^{max}$, $Y_{mi}(t-1)\leq Y_{mi}^{max}$ and $Z_{mi}(t-1)\leq Z_{mi}^{max}$. We have the following possible cases:
\begin{itemize}
\item $Y_{mi}(t-1)\in [0, Y_{mi}^{max}-A_{mi}^{max}\cdot s_m)$ or $Y_{mi}(t-1)\in [Y_{mi}^{max}- A_{mi}^{max}\cdot s_m, Y_{mi}^{max}]$;
\item $Q_{mi}(t-1)\in [0, Q_{mi}^{max}-A_{mi}^{max}\cdot s_m]$ or $Q_{mi}(t-1)\in (Q_{mi}^{max}- A_{mi}^{max}\cdot s_m, Q_{mi}^{max}]$;
\item $Z_{mi}(t-1)\in [0, Z_{mi}^{max}-\epsilon_{mi}]$ or $Z_{mi}(t-1)\in (Z_{mi}^{max}-\epsilon_{mi}, Z_{mi}^{max}]$.
\end{itemize}

We first prove the upper bounds for virtual queues, $Y_{mi}(t),\ \forall m\in \mathcal{M}, i\in \mathcal{E}$.

\vspace{1mm}\noindent
$\triangleright$ If $0\leq Y_{mi}(t-1)< Y_{mi}^{max}-A_{mi}^{max}\cdot s_m$, we have that

\vspace{-4mm}{\small
\begin{align*}
\begin{split}
\eta_{mi}(t-1)&=\max\{\min\{U'^{-1}(\frac{Y_{mi}(t-1)}{V}) / s_m,A_{mi}^{max}\},0\}\\
         &>\max\{\min\{U'^{-1}(U'(0)) / s_m,A_{mi}^{max}\},0\}=0,
\end{split}
\end{align*}}\vspace{-4mm}

\noindent according to formula (\ref{eqn:rate1}). Note that the
inequality is because that $U(\cdot)$ is differential and concave,
which means $U'^{-1}(\cdot)$ is a decreasing function.

Thus, in this case, $0<\eta_{mi}(t-1)\leq A_{mi}^{max}$. We have that

\vspace{-4mm}{\small
\begin{align*}
Y_{mi}(t)&=\max\{Y_{mi}(t-1)-r_{mi}(t-1)\cdot s_m,0\}+\eta_{mi}(t-1)\cdot s_m\\
        &\leq \max\{Y_{mi}(t-1),0\}+A_{mi}^{max}\cdot s_m\\
        &<Y_{mi}^{max},
\end{align*}}\vspace{-4mm}

\noindent according to the queueing law (\ref{eqn:virtual-queue2}).

\vspace{1mm}\noindent
$\triangleright$ If $Y_{mi}^{max}- A_{mi}^{max}\leq Y_{mi}(t-1)\leq Y_{mi}^{max}$, we have that

\vspace{-4mm}{\small
\begin{align*}
\begin{split}
\eta_{mi}(t-1)&=\max\{\min\{U'^{-1}(\frac{Y_{mi}(t-1)}{V})/s_m,A_{mi}^{max}\},0\}\\
         &\leq \max\{\min\{U'^{-1}(U'(0))/s_m,A_{mi}^{max}\},0\}=0.\\
\end{split}
\end{align*}}\vspace{-4mm}

Thus, in this case, $\eta_{mi}(t-1)\leq 0$. We have that

\vspace{-4mm}{\small
\begin{align*}
Y_{mi}(t)&=\max\{Y_{mi}(t-1)-r_{mi}(t-1)\cdot s_m,0\}+\eta_{mi}(t-1)\cdot s_m\\
        &\leq \max\{Y_{mi}(t-1),0\}\leq Y_{mi}^{max}
\end{align*}}\vspace{-4mm}

Up to now, $Y_{mi}(t)\leq Y_{mi}^{max},\ \forall m\in \mathcal{M}, i\in \mathcal{E}$ for each
time slot $t$ is proved.

Now, we discuss with the size of $Q_{mi}(t)$:

\vspace{1mm}\noindent
$\triangleright$  If $0\leq Q_{mi}(t-1) \leq Q_{mi}^{max} - A_{mi}^{max}s_m$, we have that

%\vspace{-4mm}{\small
%\begin{align*}
%\begin{split}
%r_{mi}(t-1)=A_{mi}(t-1),
%\end{split}
%\end{align*}}\vspace{-4mm}
%
%\noindent according to formula (\ref{eqn:rate2}).
%
%Thus, we further have that

\vspace{-4mm}{\small
\begin{align*}
Q_{mi}(t)=&\max\{Q_{mi}(t-1)-d_{mi}(t-1)s_m-\mu_{mi}(t-1), 0\}\\
&+r_{mi}(t-1)s_m\\
\leq &\max\{Q_{mi}^{max} - A_{mi}^{max}s_m, 0\} + A_{mi}^{max}s_m\\
= & Q_{mi}^{max},
\end{align*}}\vspace{-4mm}

\noindent according to the queuing law (\ref{eqn:packet-queue}).

\vspace{1mm}\noindent
$\triangleright$ If $Q_{mi}^{max} - A_{mi}^{max}s_m < Q_{mi}(t-1)\leq Q_{mi}^{max}$, we have that

\vspace{-4mm}{\small
\begin{align*}
Q_{mi}(t-1) > Q_{mi}^{max} - A_{mi}^{max}s_m = Y_{mi}^{max}\geq Y_{mi}(t-1).
\end{align*}}\vspace{-4mm}

Hence, we have that

\vspace{-4mm}{\small
\begin{align*}
\begin{split}
r_{mi}(t-1)=0,
\end{split}
\end{align*}}\vspace{-4mm}

\noindent according to formula (\ref{eqn:rate2}).

Thus, we further have that

\vspace{-4mm}{\small
\begin{align*}
Q_{mi}(t)=&\max\{Q_{mi}(t-1)-d_{mi}(t-1)s_m - \mu_{mi}(t-1)\}\\
&+r_{mi}(t)s_m\\
\leq& \max\{Q_{mi}^{max}, 0\} = Q_{mi}^{max}.
\end{align*}}\vspace{-4mm}

Up to now, $Q_{mi}(t)\leq Q_{mi}^{max},\ \forall m\in \mathcal{M}$
for each time slot $t$ is proved.

%-------------------------------------------------
%
%We next prove the upper bounds for packet queues, $Q_{mi}(t),\ \forall m\in \mathcal{M}, i\in \mathcal{E}$.
%
%\vspace{1mm}\noindent
%$\triangleright$ If $Q_{mi}(t-1)\in [0, Q_m^{max}-R_m\times s_m]$, we have that
%\begin{align*}
%Q_{mi}(t)&=\max\{Q_{mi}(t-1)-\mu_{mi}(t-1)-d_{mi}(t-1)\times s_m, 0\} + r_{mi}(t-1)\times s_m\\
%         &\leq Q_m^{max}-R_m\times s_m + R_m\times s_m = Q_m^{max}.
%\end{align*}
%The inequality is based on the fact that $\mu_{mi}(t-1)\in \{0,1\}$, $d_{mi}(t-1)\in [0, d_m^{max}]$ and $r_{mi}(t-1)\in [0, R_m]$.
%
%\vspace{1mm}\noindent
%$\triangleright$ If $Q_{mi}(t-1)\in (Q_m^{max}- R_m\times s_m, Q_m^{max}]$, we have that
%\begin{align*}
%Q_{mi}(t-1)+Z_{mi}(t-1)\geq V.
%\end{align*}
%Thus, the job-drop decision in time slot $t-1$ is $d_{mi}(t-1)=d_m^{max}$. Then, we have that
%\begin{align*}
%Q_{mi}(t)&=\max\{Q_{mi}(t-1)-\mu_{mi}(t-1)-d_{mi}(t-1)\times s_m, 0\} + r_{mi}(t-1)\times s_m\\
%         &\leq \max\{Q_m^{max}-d_m^{max}\times s_m, 0\} + R_m\times s_m \leq Q_m^{max}.
%\end{align*}
%Here, the second inequality comes from the fact that $d_m^{max}\geq R_m + \epsilon_m/s_m$.
%
%Up to now, we have proved that $Q_{mi}(t)\leq Q_m^{max}$.

We next prove the queue length bounds for each virtual queue $Z_{mi}(t),\ \forall m\in \mathcal{M}, i\in \mathcal{E}$.

\vspace{1mm}\noindent
$\triangleright$ If $Z_{mi}(t-1)\in [0, Z_{mi}^{max}-\epsilon_{mi}]$, we have that

\vspace{-4mm}{\small
\begin{align*}
Z_{mi}(t)=&\max\{Z_{mi}(t-1)+ \mathbf{1}_{\{Q_{mi}(t-1)>0\}}(\epsilon_{mi} - \mu_{mi}(t-1))\\
          &-d_{mi}(t-1)\times s_m-\mathbf{1}_{\{Q_{mi}(t-1)=0\}} , 0\}\\
         \leq & \max\{Z_{mi}^{max}-\epsilon_{mi} + \epsilon_{mi} , 0\} = Z_{mi}^{max}.
\end{align*}}\vspace{-4mm}

\vspace{1mm}\noindent
$\triangleright$ If $Z_{mi}(t-1)\in (Z_{mi}^{max}-\epsilon_{mi}, Z_{mi}^{max}]$, we have that

\vspace{-4mm}{\small
\begin{align*}
Q_{mi}(t-1)+Z_{mi}(t-1)> V\cdot \beta/s_m.
\end{align*}}\vspace{-4mm}

Thus, the job-drop decision in time slot $t-1$ is $d_{mi}(t-1)=d_m^{max}$ according to Eqn.~(\ref{eqn:drop-decision}). Then, we have that

\vspace{-4mm}{\small
\begin{align*}
Z_{mi}(t)=&\max\{Z_{mi}(t-1)+ \mathbf{1}_{\{Q_{mi}(t-1)>0\}}(\epsilon_{mi} - \mu_{mi}(t-1))\\
          &-d_{mi}(t-1)\times s_m-\mathbf{1}_{\{Q_{mi}(t-1)=0\}}, 0\}\\
         \leq & \max\{Z_{mi}^{max}+ \epsilon_{mi} - d_m^{max}\times s_m, 0\} \leq Z_{mi}^{max}.
\end{align*}}\vspace{-4mm}

\noindent Here, the second inequality is based on the fact that $d_{mi}^{max}\geq A_{mi}^{max} + \epsilon_{mi}/s_m$.

Thus, we also prove that $Z_{mi}(t)\leq Z_{mi}^{max}$.

\end{proof}

}

\opt{long}{
\section{Proof to Theorem \ref{theorem:delay}}\label{appendix:delay}
\begin{proof}
We prove this theorem by contradiction.

For each job type $m\in \mathcal{M}$ on each link $i\in \mathcal{E}$, the admitted jobs at time slot $t\geq 0$ is $r_{mi}(t)$ and the earliest time they can depart the queue $Q_{mi}(t)$ is $t+1$. We show that all these jobs depart (by being either scheduled or dropped) on or before $t+D_{mi}$.

Suppose this is not true, we will come to a contradiction. We must have that $Q_{mi}(\tau) >0 $ for all $\tau\in [t+1, \ldots, t+D_i]$ (otherwise, all the jobs have departed by time $t+D_{mi}$). With the queueing law in Eqn.~(\ref{eqn:virtual-queue}), we have that

\vspace{-4mm}{\small
\begin{align*}
Z_{mi}(\tau+1) =& \max\{Z_{mi}(\tau) + \epsilon_{mi} - \mu_{mi}(\tau) - d_{mi}(\tau)\times s_m, 0\}\\
           \geq & Z_{mi}(\tau) + \epsilon_{mi} - \mu_{mi}(\tau) - d_{mi}(\tau)\times s_m.
\end{align*}}\vspace{-4mm}

Summing the above over $\tau\in [t+1, \ldots, t+D_{mi}]$, we have that

\vspace{-4mm}{\small
\begin{align*}
Z_{mi}(t+D_i+1) - Z_{mi}(t+1) \geq \epsilon_{mi} \cdot D_{mi} - \sum_{\tau = t+1}^{t+D_{mi}}[\mu_{mi}(\tau) + d_{mi}(\tau)\times s_m].
\end{align*}}\vspace{-4mm}

Rearranging the above inequality and using the fact that $Z_{mi}(t+D_{mi}+1) \leq Z_{mi}^{max}$ and $Z_{mi}(t+1)\geq 0$, we have that

\vspace{-4mm}{\small
\begin{align}\label{eqn:qos-proof1}
\epsilon_{mi} \cdot D_{mi}  - Z_{mi}^{max} \leq \sum_{\tau = t+1}^{t+D_{mi}}[\mu_{mi}(\tau) + d_{mi}(\tau)\times s_m].
\end{align}}\vspace{-4mm}

Since the jobs are departing in a FIFO fashion, the jobs $r_{mi}(t)$, which arrive and are admitted at slot $t$, are placed at the end of the queue at slot $t+1$, and should be fully cleared when all the jobs backlogged in $Q_{mi}(t+1)$ have departed. That is, the last job of $r_{mi}(t)$ is scheduled or dropped on slot $t+T$ with $T>0$ as the smallest integer satisfying $\sum_{\tau = t+1}^{t+T}[\mu_{mi}(\tau)+d_{mi}(\tau)\times s_m]\geq Q_{mi}(t+1)$. Based on our assumption that not all of the $r_{mi}(t)$ jobs depart by time $t+D_{mi}$, we must have that

\vspace{-4mm}{\small
\begin{align}\label{eqn:qos-proof2}
\sum_{\tau = t+1}^{t+D_{mi}}[\mu_{mi}(\tau)+d_{mi}(\tau)\times s_m] < Q_{mi}(t+1) \leq Q_{mi}^{max}.
\end{align}}\vspace{-4mm}

Combining Eqn.~(\ref{eqn:qos-proof1}) and (\ref{eqn:qos-proof2}), we have that

\vspace{-4mm}{\small
\begin{align*}
&\epsilon_{mi} \cdot D_{mi}  - Z_{mi}^{max} < Q_{mi}^{max}\\
\Rightarrow & \epsilon_{mi} < \frac{Q_{mi}^{max}+ Z_{mi}^{max}}{D_{mi}}.
\end{align*}}\vspace{-4mm}

This contradicts with the given fact that $\epsilon_{mi} = \frac{Q_{mi}^{max}+ Z_{mi}^{max}}{D_{mi}}$. Hence, we have proved that each job of type $m$ on link $i$ is either scheduled or dropped with Algorithm \ref{alg:all} before its maximum delay $D_{mi}$, if we set $\epsilon_{mi} = \frac{Q_{mi}^{max}+ Z_{mi}^{max}}{D_{mi}}$.

%This theorem can be proved similarly with that of Lemma 1 in \cite{neely-infocom11}.
\end{proof}

}

\opt{long}{
\section{Proof to Theorem \ref{theorem:delta-weight}}\label{appendix:delta-weight}
We first prove the correctness of Algorithm \ref{alg:CSMA} by showing that the generated transmission schedules are collision-free. We prove the optimality by modeling the transmission scheduling decisions as a Discrete-Time Markov Chain (DTMC), which is next proved to be reversible. Hence, we derive the stationary distribution for each collision-free transmission scheduling decision, based on which we evaluate the achievable value for $\Phi_3(t)$ in expectation with Algorithm \ref{alg:CSMA}.

Let $\chi(t) = \{x_i(t)|x_i(t)=1, \forall i\in \mathcal{E}\}$ be a scheduling decision in time slot $t$, and $\Lambda$ be the set of all collision-free scheduling decisions. Denote $\mathcal{Z}$ be the set of all possible control schedules of $\mathbf{z}(t)$. It is clear that $\mathcal{Z}\subseteq \Lambda$. Let $\rho(\mathbf{z}(t))> 0$ be the probability of selecting $\mathbf{z}(t)$ as the control schedule. We have that $\sum_{\mathbf{z}(t)\in \mathcal{Z}}\rho(\mathbf{z}(t)) = 1$.

%We denote $\mathcal{C}(\chi(t))$ as the set of interfering link scheduling variables of link scheduling decision $\chi(t)$.

\begin{lemma}\label{lemma:feasible}
\emph{If the scheduling decision in time slot $t-1$ and the control schedule in time slot $t$ are both collision free, \emph{i.e.}, $\chi(t-1)\in \Lambda$ and $\mathbf{z}(t)\in \Lambda$, we have that the link scheduling in time slot $t$ with Algorithm \ref{alg:CSMA} is also collision free, \emph{i.e.},  $\chi(t)\in \Lambda$.}
\end{lemma}

\begin{proof}
A link scheduling decision $\chi(t)$ is collision-free if and only if $\forall x_i(t)\in \chi(t)$, we have $x_j(t)=0$, $\forall j \in \mathcal{C}_i$.

Consider any $x_i(t)\in \chi(t)$. If $x_i(t)\not\in \mathbf{z}(t)$, we have that $x_i(t-1) = x_i(t) = 1$ based on Algorithm \ref{alg:CSMA}, which means $x_i(t-1)\in \chi(t-1)$. Since $\chi(t-1)$ is collision-free, we know that $x_j(t-1) = 0$, $\forall j\in \mathcal{C}_i$. Then, we can discuss the value of $x_j(t)$ as follows.
\begin{itemize}
\item If $x_j(t)\not\in \mathbf{z}(t)$, we know that $x_j(t) = x_j(t-1) = 0$ based on Algorithm \ref{alg:CSMA}.
\item If $x_j(t) \in \mathbf{z}(t)$, we have $x_j(t) = 0$ since $x_i(t-1)\in \chi(t-1)$ and $i \in \mathcal{C}_j$.
\end{itemize}

If $x_i(t)\in \mathbf{z}(t)$, we have that $x_i(t)\in \chi(t)$ only if $x_j(t-1) = 0$, $\forall j \in \mathcal{C}_i$. Since $x_i(t)\in \mathbf{z}(t)$ and $\mathbf{z}(t)$ is collision-free, we know that $\mathcal{C}_i \cap \mathbf{z}(t) = \emptyset$. Hence, $x_j(t) = x_j(t-1) = 0$.

Thus, we prove this lemma by showing that $\forall x_i(t)\in \chi(t)$, we have $x_j(t)=0$, $\forall j \in \mathcal{C}_i$.

\end{proof}

\begin{lemma}\label{lemma:transit}
\emph{A link scheduling decision $\chi \in \Lambda$ can transit to a link scheduling decision $\chi'\in \Lambda$ if and only if $\chi \cup \chi'\in \Lambda$ and there exists a control schedule $\mathbf{z}\in \mathcal{Z}$ such that}

\vspace{-4mm}{\small
\begin{align*}
\chi \bigtriangleup \chi' = (\chi / \chi')\cup (\chi' / \chi) \subseteq \mathbf{z},
\end{align*}}\vspace{-4mm}

\noindent \emph{and the transition probability from $\chi$ to $\chi'$ is}

\vspace{-4mm}{\small
\begin{align}\label{eqn:transitionrate}
P(\chi, \chi') = & \sum_{\mathbf{z}\in \mathcal{Z}: \chi \bigtriangleup \chi' \subseteq \mathbf{z}} \rho(\mathbf{z}) \left( \prod_{x_i\in \chi / \chi'}1-p_i\right) \left( \prod_{x_i\in \chi' / \chi}p_i\right)\notag\\
& \left( \prod_{x_i\in \mathbf{z} \cap (\chi \cap \chi')}p_i\right)\left( \prod_{x_i\in \mathbf{z}/ (\chi\cup \chi')/ \mathcal{C}(\chi \cup \chi')}1-p_i\right)
\end{align}}\vspace{-4mm}

\end{lemma}

\begin{proof}
we first prove the necessity and then the sufficiency.

\noindent \emph{Necessity}: Suppose $\chi$ is the current decision in time slot $t$ and $\chi'$ is the next decision in slot $t+1$. $\chi / \chi' = \{x_i| x_i(t) = 1, x_i(t+1) = 0\}$ is the set of link scheduling variables that change their state from 1 to 0. $\chi' / \chi = \{x_i| x_i(t) = 0, x_i(t+1) = 1\}$ is the set of link scheduling variables that change their state from 0 to 1.

Based on Algorithm \ref{alg:CSMA}, we have that a link scheduling variable can change its state only if it is included in the control schedule $\mathbf{z}$. Therefore, $\chi$ can transit to $\chi'$ only if there exists a control schedule $m\in \Lambda$ such that the symmetric difference $\chi \bigtriangleup \chi' = (\chi / \chi')\cup (\chi' / \chi) \subseteq \mathbf{z}$. In addition, we have $\chi \cup \chi' = (\chi / \chi') \cup (\chi' / \chi)\cup (\chi\cap \chi') \in \Lambda$, since $(\chi \cap \chi')\cup (\chi / \chi') = \chi \in \Lambda$, $(\chi \cap \chi')\cup (\chi' / \chi) = \chi' \in \Lambda$, and $(\chi / \chi') \cup (\chi' / \chi) = \chi \bigtriangleup \chi'$.

\noindent \emph{Sufficiency}: Suppose $\chi \cup \chi' \in \Lambda$ and there is an $\mathbf{z}\in \Lambda$ such that $\chi \bigtriangleup \chi' \subseteq \mathbf{z}$. Given $\mathbf{z}$ is selected randomly, we can calculate the probability for $\chi$ to transit to $\chi'$ by dividing the variables in $\mathbf{z}$ in 5 cases as follows.
\begin{itemize}
\item $x_i(t) \in \chi / \chi'$: Variable $x_i(t)$ is decided to change its state from 1 to 0, which happens with probability $1 - p_i$ with Algorithm \ref{alg:CSMA}.
\item $x_i(t) \in \chi' / \chi$: Variable $x_i(t)$ is decided to change its state from 0 to 1, which occurs with probability $p_i$.
\item $x_i(t) \in \mathbf{z} \cap (\chi \cap \chi')$: Variable $x_i(t)$ is decided to keep the sate 1, which occurs with probability $p_i$.
\item $x_i(t) \in \mathbf{z} \cap \mathcal{C}(\chi)$: Variable $x_i(t)$ has to keep its state 0. This occurs with probability 1.
\item $x_i(t) \in \mathbf{z} / (\chi \cup \chi')/ \mathcal{C}(\chi)$: Variable $x_i(t)$ decides to keep its state 0, which happens with probability with $1 - p_i$
\end{itemize}

Note that $\mathbf{z}\cap \mathcal{C}(\chi' / \chi) = \emptyset$ since $\chi'/ \chi\subseteq \mathbf{z}$, we have that $\mathbf{z}/(\chi \cup \chi')/\mathcal{C}(\chi) = \mathbf{z} / (\chi \cup \chi') / \mathcal{C}(\chi \cup \chi')$. As each variable in $\mathbf{z}$ is decided independently of each other, we can multiply these probabilities together. Summing over all possible control schedules, we can get the overall transition probability from $\chi$ to $\chi'$ as in Eqn.~(\ref{eqn:transitionrate}).
\end{proof}

\begin{lemma}\label{lemma:reversible}
\emph{A necessary and sufficient condition for the DTMC of the link scheduling decisions to be irreducible and aperiodic is}

\vspace{-4mm}{\small
\begin{align*}
\bigcup_{\mathbf{z}(t)\in \mathcal{Z}}\mathbf{z}(t) = \{x_i|\forall i \in \mathcal{E}\},
\end{align*}}\vspace{-4mm}

\noindent \emph{and in this case DTMC is reversible and has the following stationary distribution,}

\vspace{-4mm}{\small
\begin{align}
&\pi(\chi) = \frac{1}{H}\prod_{x_i\in \chi}\frac{p_i}{1-p_i},\label{eqn:stationary}\\
&H = \sum_{\chi\in \Lambda}\prod_{x_i\in \chi}\frac{p_i}{1-p_i},\label{eqn:normalization}.
\end{align}}\vspace{-4mm}
\end{lemma}

\begin{proof}
We first prove the necessity and sufficiency condition for the DTMC to be irreducible and aperiodic, and next verify the reversibility and stationary distribution.

\noindent \emph{Necessity}: Suppose $\bigcup_{\mathbf{z}(t)\in \mathcal{Z}}\mathbf{z}(t) \neq \{x_i|\forall i\in \mathcal{E}\}$. Let $x_j(t)\not\in \bigcup_{\mathbf{z}(t)\in \mathcal{Z}}\mathbf{z}(t)$. Then, we have that, from the initial state of link scheduling decision, \emph{i.e.}, $\emptyset$, the DMTC will never reach a collision-free decision including $x_j(t)$. The necessity is proved.

\noindent \emph{Sufficiency}: If $\bigcup_{\mathbf{z}(t)\in \mathcal{Z}}\mathbf{z}(t) = \{x_i|\forall i \in \mathcal{E}\}$, with Lemma \ref{lemma:transit}, we have that the initial decision $\emptyset$ can reach any other collision-free decision $\chi \in \Lambda$ with positive probability in a finite number of steps, and vice versa. To sum up, the DTMC is irreducible and aperiodic.

If allocation decision $\chi$ can transit to decision $\chi'$, we can verify that Eqn.~(\ref{eqn:stationary}) satisfies the balance equation,

\vspace{-4mm}{\small
\begin{align}\label{eqn:balance-eqn}
\pi(\chi) P(\chi, \chi') =& \frac{1}{H}\sum_{\mathbf{z}\in \mathcal{Z}: \chi \bigtriangleup \chi' \subseteq \mathbf{z}} \rho(\mathbf{z}) \left(\frac{\prod_{x_i\in \chi \cup \chi'} p_i}{\prod_{x_j\in \chi\cap \chi'} 1- p_j}\right)\notag\\
& \times \left( \prod_{x_i\in \mathbf{z}\cap (\chi\cap \chi')}p_i\right) \left(\prod_{x_i\in \mathbf{z}/ (\chi \cup \chi')/ \mathcal{C}(\chi \cup \chi')} 1- p_i\right)\notag\\
=&\pi(\chi')P(\chi', \chi).
\end{align}}\vspace{-4mm}

Therefor, the DTMC is reversible and Eqn.~(\ref{eqn:stationary}) is the stationary distribution \cite{kelly-book79}.

\end{proof}

Finally, we prove Theorem \ref{theorem:delta-weight} based on the above lemmas.

\begin{proof}
Given any $\delta$ and $\theta$ with $0 < \delta, \theta < 1$. Let $\Phi_3^*(t) = \max_{\chi \in \Lambda}\Phi_3(t)$. We define

\vspace{-4mm}{\small
\begin{align*}
\mathcal{\chi} = \left\{ \chi \in \Lambda| \Phi_3(t) < (1-\delta)\Phi_3^*(t) \right\}.
\end{align*}}\vspace{-4mm}

As the DTMC has the stationary distribution in Eqn.~(\ref{eqn:stationary}), we have that

\vspace{-4mm}{\small
\begin{align}\label{eqn:optimal-condition1}
\pi(\mathcal{\chi}) &= \sum_{\chi \in \mathcal{\chi}}\pi(\chi) = \sum_{\chi \in \mathcal{\chi}}\frac{e^{\sum_{x_i(t)\in \chi}w_i(t)}}{H}\notag\\
&\leq \frac{|\mathcal{\chi}|e^{1-\delta} \Phi_3^*(t)}{H}< \frac{2^{|\mathcal{E}|}}{e^{\delta \Phi_3^*(t)}},
\end{align}}\vspace{-4mm}

\noindent where the secondary inequality comes from the fact that $|\mathcal{\chi}|\leq |\Lambda| \leq 2^{|\mathcal{E}|}$, and $H > e^{\max_{\chi\in \Lambda}\sum_{x_i\in \chi}w_i(t)} = e^{\Phi_3^*(t)}$. Hence, if

\vspace{-4mm}{\small
\begin{align}\label{eqn:optimal-condition2}
\Phi_3^*(t) > \frac{1}{\delta}\left(|\mathcal{E}| \log 2 + \log \frac{1}{\theta}\right),
\end{align}}\vspace{-4mm}

\noindent we could have that $\pi(\mathcal{\chi})< \theta$. Since $\Phi_3^*(t)$ is a continuous, nondecreasing function of the packet queues and QoS virtual queues, \emph{i.e.}, $\Gamma(t) = \{Q_{mi}(t), Z_{mi}(t)|\forall m\in \mathcal{M}, i \in \mathcal{E}\}$, we can further have that, with $\lim_{\|\Gamma(t)\|\rightarrow \infty} \Phi_3^*(t) = \infty$, there exists a constant value $B_{\Gamma}$ such that inequality (\ref{eqn:optimal-condition2}) holds so that $\pi(\mathcal{\chi})< \delta$ whenever $\|\Gamma(t)\| > B$.

%According to Lemma \ref{lemma:q-bound-s}, the value of $\Gamma(t)$ is proportional to $V$. By scaling up $V\rightarrow \infty$, we have that $\Phi_3^*(t)\rightarrow \infty$ such that $\delta\rightarrow 0$ and $\theta \rightarrow 0$ with Eqn.~(\ref{eqn:optimal-condition1}) and (\ref{eqn:optimal-condition2}). Since the DTMC generate a value of $\Psi(t)$ in $(1-\delta)\Psi^*(t)$ with a probability of at least $1-\theta$, we have that, with $V\rightarrow \infty$, Algorithm \ref{alg:CSMA} results in an expected value of $\Psi(t)$ arbitrarliy close to $\Psi^*(t)$.

\end{proof}
}

\end{appendices}

\end{document}